%% file: zhou_ordinal_arxiv.tex
\documentclass[USenglish,cleveref, autoref, thm-restate]{lipics-v2019}
\input{header_martin.tex}

\title{The Zhou Ordinal of Labelled Markov Processes over Separable Spaces} %

\titlerunning{Zhou Ordinal of Separable Labelled Markov Processes} %

\author{Martín Santiago Moroni}{Universidad Nacional de C\'ordoba. 
    Facultad de Matem\'atica, Astronom\'{\i}a,  F\'{\i}sica y
    Computaci\'on. \and     Centro de Investigaci\'on y Estudios de Matem\'atica (CIEM-FAMAF),
    Conicet. C\'ordoba. Argentina}{moroni@famaf.unc.edu.ar}{}{}%

\author{Pedro Sánchez Terraf}{Universidad Nacional de C\'ordoba. 
    Facultad de Matem\'atica, Astronom\'{\i}a,  F\'{\i}sica y
    Computaci\'on. \and     Centro de Investigaci\'on y Estudios de Matem\'atica (CIEM-FAMAF),
    Conicet. C\'ordoba. Argentina
      \and \url{https://cs.famaf.unc.edu.ar/~pedro/}}{sterraf@famaf.unc.edu.ar}{https://orcid.org/0000-0003-3928-6942}{}

\authorrunning{M.\,S. Moroni and P. Sánchez Terraf} %

\Copyright{Martín Santiago Moroni and Pedro S\'anchez Terraf} %



\category{} %

\relatedversion{} %

\supplement{}%

\funding{Supported by Secyt-UNC project 33620180100465CB.}%

\acknowledgements{We are very grateful for the referee's useful 
suggestions and detailed reading. The second author wants to thank 
Chunlai Zhou for early discussions on this material during Dagstuhl 
Seminar 12411 on Coalgebraic Logics.}%

\nolinenumbers %

\hideLIPIcs  %

\EventEditors{John Q. Open and Joan R. Access}
\EventNoEds{2}
\EventLongTitle{42nd Conference on Very Important Topics (CVIT 2016)}
\EventShortTitle{CVIT 2016}
\EventAcronym{CVIT}
\EventYear{2016}
\EventDate{December 24--27, 2016}
\EventLocation{Little Whinging, United Kingdom}
\EventLogo{}
\SeriesVolume{42}
\ArticleNo{23}

\begin{document}

\maketitle

\begin{abstract}
  There exist two notions of equivalence of behavior between states of
  a Labelled Markov Process (LMP): state bisimilarity and event
  bisimilarity. The first one can be considered as an appropriate
  generalization to continuous spaces of Larsen and Skou's
  probabilistic bisimilarity, while the second one is characterized by
  a natural logic. C.~Zhou expressed state bisimilarity as the greatest
  fixed point of an operator $\Op$, and thus introduced an ordinal
  measure of the discrepancy between it and event bisimilarity. We
  call this ordinal the \emph{Zhou ordinal} of $\lmp{S}$,
  $\Zh(\lmp{S})$. When $\Zh(\lmp{S})=0$, $\lmp{S}$ satisfies the
  Hennessy-Milner property. The second author proved the existence of
  an LMP $\lmp{S}$ with $\Zh(\lmp{S}) \geq 1$ and Zhou showed that there
  are LMPs having an infinite Zhou ordinal.  In this paper we show
  that there are LMPs $\lmp{S}$ over separable metrizable spaces
  having arbitrary large countable $\Zh(\lmp{S})$ and that it is
  consistent with the axioms of $\ZFC$ that there is such a process
  with an uncountable Zhou ordinal. 
\end{abstract}

\input{intro}

\input{preliminaries}
\input{operators_O_G}

\input{zhou}

\input{cofinality}

\input{example}

\input{conclusion}

\input{zhou_ordinal.bbl}

\end{document}

%% file: header_martin.tex
\usepackage{amsmath}
\usepackage{amsthm}
\usepackage{amsfonts}
\usepackage{amssymb}
\usepackage{verbatim}
\usepackage{enumerate}
\usepackage{stmaryrd}
\usepackage[numbers]{natbib}
\usepackage{mathtools}
\usepackage{ulem}
\usepackage{cancel} %
\usepackage[utf8]{inputenc} %
\usepackage{xcolor}
\usepackage[all]{xy} \SilentMatrices
\usepackage{fge} %
\usepackage{graphicx} 
\newcommand\rsetminus{\mathbin{\mathpalette\rsetminusaux\relax}}
\newcommand\rsetminusaux[2]{\mspace{-4mu}
	\raisebox{\rsmraise{#1}\depth}{\rotatebox[origin=c]{-20}{$#1\smallsetminus$}}
	\mspace{-4mu}
}
\newcommand\rsmraise[1]{%
	\ifx#1\displaystyle .8\else
	\ifx#1\textstyle .8\else
	\ifx#1\scriptstyle .6\else
	.45%
	\fi
	\fi
	\fi}
\newcommand{\minus}{\rsetminus}

\newcommand{\defi}{\mathrel{\mathop:}=} %
\newcommand{\A}{\mathcal{A}}
\newcommand{\R}{\mathbb{R}}

\newcommand{\Q}{\mathbb{Q}}

\newcommand{\Borel}{\mathcal{B}}
\newcommand{\Power}{\mathcal{P}}
\renewcommand{\emptyset}{\varnothing}
\DeclareMathOperator{\comp}{\mathrm{comp}}

\newcommand{\Rel}{\mathcal{R}}
\newcommand{\Op}{\mathcal{O}}
\newcommand{\G}{\mathcal{G}}
\newcommand{\Zh}{\mathfrak{Z}}

\newcommand{\Logic}{\mathcal{L}}
\newcommand{\sem}[1]{{\llbracket #1 \rrbracket}} %
\newcommand{\pos}[1]{{\langle #1 \rangle}} %
\newcommand{\posg}[1]{{\langle #1 \rangle_{>q}}} %
\newcommand{\posleq}[1]{{\langle #1 \rangle_{\leq q}}} %

\DeclareMathOperator{\cf}{cf}

\newcommand{\lmp}[1]{\mathbb{#1}}
\newcommand{\axiomas}[1]{\mathit{#1}}
\newcommand{\MA}{\axiomas{MA}}
\newcommand{\CH}{\axiomas{CH}}
\newcommand{\ZFC}{\axiomas{ZFC}}

\newcommand{\sm}{\minus}
\newcommand{\sbq}{\subseteq}
\newcommand{\seplmp}{\mathcal{S}}
\newcommand{\card}[1]{{\left|#1\right|}}

\newcommand{\leb}{\mathfrak{m}}

\newcommand{\Inter}{I}

\newcommand{\oet}{\triangledown} %
\newcommand{\lpet}{\dag} %

\DeclareMathSymbol{\rest}{\mathord}{AMSa}{"16} 

\newtheorem*{example*}{Example}

%% file: intro.tex
\section{Introduction}\label{sec:intro}

Equivalence of behavior, or \emph{bisimilarity} in any of its flavors
is a fundamental concept in the study of processes, logic, and many other
areas of Computer Science and Mathematics. In the case of discrete
(countable) processes, many formalizations of the concept result to be
equivalent and it can be completely described by using some form of
modal logic--- the well-known Hennessy-Milner property.

As soon as one leaves the realm of discrete processes, the question of
defining and characterizing behavior turns into a problem with various
(sometimes unexpected) mathematical edges. For the case of
\emph{labelled Markov processes (LMP)} \cite{Desharnais}, 
the first issue to be taken care of
is that the concept of probability and measure cannot be defined for
all subsets of the state space. Hence the complexity of  state
spaces (in the sense of Descriptive Set Theory \cite{Kechris}) plays an important
role. 

A notable consequence is that an LMP admits two generally different
notions of equivalence of behavior between its states: \emph{state}
bisimilarity and \emph{event}
bisimilarity. The first one can be considered as an appropriate
generalization to continuous spaces of Larsen and Skou's
probabilistic bisimilarity. On the other hand, event bisimilarity can
be characterized by 
a very simple and natural modal logic $\Logic$ defined by the
following grammar,
\[
\phi \defi \top \mid \phi_1 \land \phi_2 \mid \langle a \rangle_{>q}\phi,
\]
where $a$ ranges over possible actions of the interpreting LMP and $q$
over rationals between $0$ and $1$; the formula $\langle a
\rangle_{>q}\phi$ holds on states at which the probability of 
reaching another state satisfying $\phi$ after an $a$ transition is 
greater than $q$.

Despite its simplicity, this
logic also characterizes state bisimilarity for wide classes of LMPs
(thus, the two types of bisimilarities coincide).
Desharnais, Edalat and Panangaden
\cite{DEP} showed (building on Edalat's categorical result
\cite{Edalat}) that the category of \emph{generalized} LMP over
analytic state spaces has the Hennessy-Milner property with respect to
$\Logic$. This result
was later strengthened by Doberkat
\cite{Doberkat:2005:SSR:1089905.1089907} in that it applies to the
original category of LMP. Recently, Pachl and the second author
extended the result to LMP over universally measurable state spaces
\cite{2017arXiv170602801P}.

But if regularity assumptions on the state spaces are omitted,
the Hennessy-Milner property is lost (see \cite{Pedro20111048} by the
second author). It is therefore of interest to
understand how state bisimilarity differs from the event one for LMP
over general measurable spaces. Zhou proposed in
\cite{DBLP:journals/entcs/Zhou13} one way to quantify this difference,
by expressing state bisimilarity as the greatest fixed point of an
operator $\Op$ and pointed out an LMP for which more than $\omega$
iterates of $\Op$ are needed to reach it.

In this paper, we study the operator $\Op$ in a general setting, a
dual version $\G$ of it, and the hierarchy of relations and
$\sigma$-algebras respectively induced by them. We then define the
\emph{Zhou ordinal} $\Zh(\lmp{S})$ of an LMP $\lmp{S}$ to be the number of
iterates needed to reach state bisimilarity when one starts from the
event one. After reviewing some basic material in
Section~\ref{sec:preliminaries}, we develop the general theory of the
operators $\Op$ and $\G$ in
Section~\ref{sec:operators-op-g}. In
Section~\ref{sec:zhou-ordinal} we focus on the class
$\seplmp$ of
LMP over separable metrizable spaces, ``separable LMP''
for short, and
the supremum of the Zhou ordinals of such processes,
$\Zh(\seplmp)$. One of our main results is that  $\Zh(\seplmp)$ is a limit
ordinal of uncountable cofinality (and hence at least $\omega_1$).
In Section~\ref{sec:example}, 
we construct a family of LMPs
$\{\lmp{S}(\beta) \mid \beta\leq\omega_1\}$ having
$\Zh(\lmp{S}(\beta))= \beta$ when $\beta$ is a limit ordinal; these
processes are separable for 
countable $\beta$. We also discuss the
consistency with the axioms of set theory that the bound $\omega_1$ is actually
attained by a separable LMP. Finally, some further directions are pointed out
in Section~\ref{sec:conclusion}.

%% file: preliminaries.tex
\section{Preliminaries}
\label{sec:preliminaries}

An algebra over a set $S$ is a nonempty family of subsets of $S$ 
closed under finite unions and complementation. It is a 
\emph{$\sigma$-algebra} if it is also closed under countable unions. 
Given an arbitrary family $\mathcal{U}$ of subsets of $S$, we use 
$\sigma(\mathcal{U})$ to denote the least $\sigma$-algebra over $S$ 
containing $\mathcal{U}$.
Let $(S,\Sigma)$ be a \emph{measurable space}, i.e., a set $S$ with a 
$\sigma$-algebra $\Sigma$ over $S$. We say that $(S,\Sigma)$ (or 
$\Sigma$) is countably generated if there is some countable family 
$\mathcal{U}\sbq \Sigma$ such that $\Sigma=\sigma(\mathcal{U})$. A 
\emph{subspace} of the measurable space $(S,\Sigma)$ consists of a 
subset $Y\subseteq S$ with the \emph{relative $\sigma$-algebra} 
$\Sigma\restriction Y\defi\{A\cap Y\mid A\in \Sigma\}$. Notice that if 
$\Sigma=\sigma(\mathcal{U})$, then $\Sigma\restriction 
Y=\sigma(\mathcal{U}\restriction Y)$.
If $(S_1,\Sigma_1),(S_2,\Sigma_2)$ are two measurables spaces, we say 
that $f:S_1\to S_2$ is 
\emph{$(\Sigma_1,\Sigma_2)$-measurable} if $f^{-1}(A)\in \Sigma_1$ 
for all $A\in \Sigma_2$. 

Assume now that $V\sbq S$. We will use $\Sigma_V$ to denote 
$\sigma(\Sigma \cup \{V\})$, the extension of $\Sigma$ by the set V. 
It is immediate that $\Sigma_V=\{(B_1\cap V)\cup(B_2\cap V^c)\mid 
B_1,B_2\in \Sigma\}$. It is obvious that if $\Sigma$ is 
countably generated so is $\Sigma_V$.
The \emph{sum} of two measurable spaces $(S_1,\Sigma_1)$ and 
$(S_2,\Sigma_2)$ is $(S_1\oplus S_2,\Sigma_1\oplus \Sigma_2)$, with 
the following abuse of notation: $S_1\oplus S_2$ is the disjoint 
union (direct sum qua sets) and $\Sigma_1\oplus\Sigma_2\defi\{Q_1\oplus 
Q_2\mid Q_i \in \Sigma_i\}$.
If $Y$ is a topological space, $\Borel(Y)$ will denote the 
$\sigma$-algebra generated by the open sets in $Y$, hence 
$(Y,\Borel(Y))$ 
is a measurable space, \emph{the Borel space} of $Y$.
We say that a family of sets $\mathcal{F}\sbq \Power(S)$ 
\emph{separates
points} if for every pair of distinct points $x,y$ in $S$, there is some
$A\in \mathcal{F}$ with $x\in A$ and $y \notin A$. 
We have the following proposition
\begin{proposition}[{\cite[Prop. 
12.1]{Kechris}}]\label{prop:meas-space-equivalences}
  Let $(S,\Sigma)$ be a measurable space. The following are 
  equivalent:
  \begin{enumerate}
    \item $(S,\Sigma)$ is isomorphic to some $(Y,\Borel(Y))$, where 
    $Y$ is separable metrizable.
    \item $(S,\Sigma)$ is isomorphic to some $(Y,\Borel(Y))$ for 
    $Y\sbq [0,1]$.  
    \item $(S,\Sigma)$ is countably generated and separates points.
  \end{enumerate}
\end{proposition}

A class $\mathcal{M}$ of subsets of $S$ is \emph{monotone} if it is 
closed under the formation of monotone unions and intersections. 
Halmos' Monotone Class Theorem will be frequently used in this work. 

\begin{theorem}[{\cite[Thm.~3.4]{billingsley}}]
  If $\mathcal{F}$ is an algebra of sets and $\mathcal{M}$ is a 
  monotone class, then $\mathcal{F}\subseteq \mathcal{M}$ implies 
  $\sigma(\mathcal{F})\subseteq \mathcal{M}$.
\end{theorem}

Given a measurable space $(S,\Sigma)$, a \emph{subprobability 
measure} on $S$ is a $[0,1]$-valued set function $\mu$ defined on 
$\Sigma$ such that $\mu(0)=0$ and for any pairwise disjoint 
collection $\{A_n\mid n\in \omega\}\sbq \Sigma$, we have 
$\mu(\bigcup_{n\in \omega}A_n)=\sum_{n\in \omega}\mu(A_n)$. In 
addition, for \emph{probability} measures we require $\mu(S)=1$. If 
$\Sigma\sbq \Sigma'$ and $\mu,\mu'$ are measures defined on 
$(S,\Sigma),(S,\Sigma')$ respectively, we say that \emph{$\mu'$ 
extends $\mu$} to $(S,\Sigma')$ when $\mu'\restriction \Sigma=\mu$.
A key idea in the construction of examples is the possibility of 
extending a measure in the following particular way:
\begin{theorem}\label{thm:measure-extension}
  Let $\mu$ be a finite measure defined in $(S,\Sigma)$ and let $V 
  \sbq S$ be a non-$\mu$-measurable set. Then there are extensions 
  $\mu_0$ and $\mu_1$ of $\mu$ to $\Sigma_V$ such that 
  $\mu_0(V)\neq 
  \mu_1(V)$.
\end{theorem}

\begin{definition} 
  A \emph{Markov kernel} on $(S,\Sigma)$ is a 
  function $\tau:S\times \Sigma \rightarrow [0,1]$ such that 
  for each fixed $s \in S$, $\tau(s,\cdot):\Sigma \rightarrow [0,1]$
  is a  subprobability measure, and for each fixed set  $X \in
  \Sigma$, $\tau(\cdot,X):S \rightarrow [0,1]$ is  $(\Sigma,\Borel[0,1])$-measurable.
\end{definition}

These kernels  will play the role of transition functions in
the processes we define next. Let $L$ be a countable set.

\begin{definition}\label{def LMP}
  A \emph{labelled  Markov process (LMP)} with \emph{label set} $L$ is
  a  triple $\lmp{S}=(S,\Sigma,\{\tau_a
  \mid a \in L\})$, where $S$ is a set of \emph{states}, $\Sigma$ is a
  $\sigma$-algebra over   $S$, and for each  $a\in L$, $\tau_a:S
  \times \Sigma \rightarrow  [0,1]$ is a  Markov kernel. An LMP is
  said to be \emph{separable} if its state space is countably generated
  and separates points.
\end{definition}
By Proposition~\ref{prop:meas-space-equivalences}, the restriction to 
separable LMP is equivalent to
studying processes whose state space is a subset of Euclidean
space.

\begin{example}\label{exm:lmp-U}
  We will now present the LMP $\lmp{U}$, which was introduced (under 
  the name $\mathbf{S_3}$) in \cite{Pedro20111048}. This will be an 
  important example throughout this paper to illustrate concepts and 
  motivate constructions. From this point onwards, $\Inter$ will 
  denote the open interval $(0,1)$, $\leb$ will be the Lebesgue 
  measure on $\Inter$ and $\Borel_V$ will be the $\sigma$-algebra 
  $\sigma(\Borel(\Inter) \cup \{V\})$, where $V$ is a Lebesgue 
  non-measurable subset of $\Inter $. By 
  Theorem~\ref{thm:measure-extension} we have two extensions $\leb_0$ 
  and $\leb_1$ of $\leb$ such that $\leb_0(V)\neq\leb_1(V)$. Also, 
  let $\{q_n\}_{n\in \omega}$ be an enumeration of the rationals in 
  $\Inter$ and   define $B_n \defi (0,q_n)$; hence   $\{B_n\mid n\in 
  \omega\}$ is a countable generating family of $\Borel(\Inter )$.
  
  Let $s,t,x \notin \Inter $ be mutually distinct; we may view 
  $\leb_0$ and $\leb_1$ as measures defined on the sum $\Inter \oplus 
  \{s,t,x\}$, supported on $\Inter $. The label set will be $L\defi 
  \omega \cup \{\infty\}$. Now define $\lmp{U}=(U,\Upsilon,\{\tau_n 
  \mid n\in L\})$ such that
  \begin{align*}
    (U,\Upsilon)& \defi (\Inter \oplus\{s,t,x\},\Borel_V\oplus 
    \Power(\{s,t,x\})), \\
    \tau_n(r,A) &\defi \chi_{B_n}(r)\cdot \delta_x(A),\\  
    \tau_\infty(r,A)&\defi \chi_{\{s\}}(r)\cdot   
    \leb_0(A)+\chi_{\{t\}}(r)\cdot \leb_1(A)
  \end{align*}
  when $n\in \omega$ and $A\in \Upsilon$. This defines an LMP since 
  for all $r$, $0\leq
  \chi_{B_a}(r) \leq 1$ and  $0\leq
  \chi_{\{s\}}(r) + \chi_{\{t\}}(r)  \leq 1$ and we infer 
  measurability because $\tau_l(\cdot,A)$ is
  always a linear combination of measurable functions.
  
  The dynamics of this process goes intuitively as follows: The 
  states $s$ and $t$  can only make an $\infty$-labelled transition  
  to a ``uniformly distributed'' state in $\Inter $, but they 
  disagree on the probability of reaching $V\sbq \Inter$. Then, each 
  point of $B_n\subseteq \Inter$ can  make an $n$-transition to $x$. 
  Finally, $x$ can make no transition at all (see 
  Figure~\ref{fig:lmp-U}).
  \begin{figure}
    \begin{center}
      \includegraphics[scale=0.08]{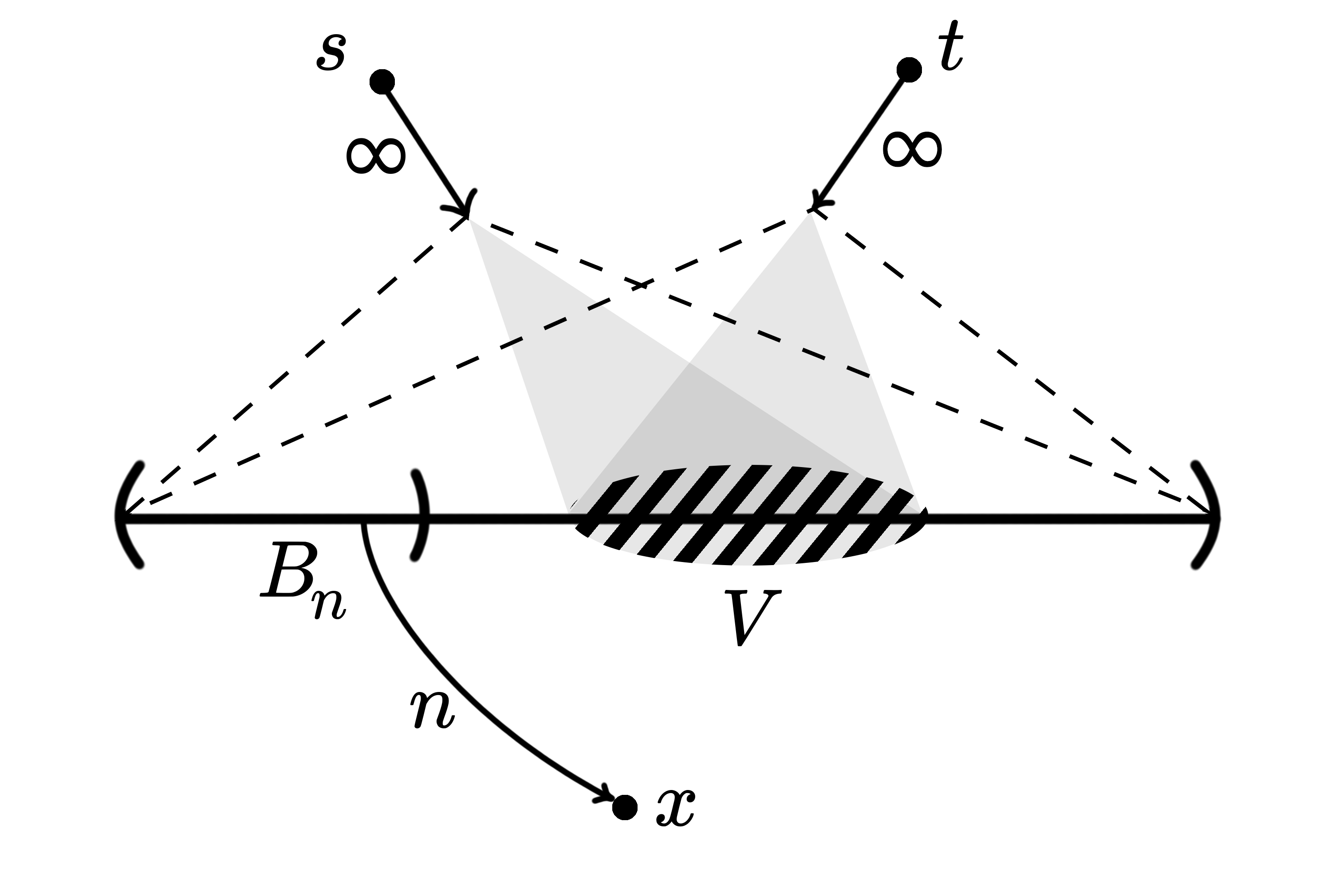}
    \end{center}
    \caption{The LMP $\lmp{U}$.}\label{fig:lmp-U}
  \end{figure}
\end{example}

For  $R$ a symmetric relation over $S$, we say that  $A 
\subseteq S$
is \emph{$R$-closed} if
$\{s\in S \mid \exists x\in A \; x\mathrel{R}s\} \subseteq A$. %
If $\Gamma\subseteq \Power(S)$, we denote by
$\Gamma(R)$ the family of all 
$R$-closed sets in $\Gamma$. Note that if
$\Gamma$ is a  $\sigma$-algebra then $\Gamma(R)$ is a 
sub-$\sigma$-algebra of $\Gamma$.  We also define a new relation
$\Rel(\Gamma)$ consisting of all pairs $(s,t)$ such that 
$\forall A \in \Gamma \; (s \in A \leftrightarrow t \in A)$. 

\begin{definition} 
  Fix an LMP $\lmp{S}=(S,\Sigma,\{\tau_a\mid a\in L\})$. A
  \emph{state bisimulation} $R$ on $\lmp{S}$ is a 
  symmetric relation on
  $S$ such that $\forall a\in L \; \tau_a(s,C)=\tau_a(t,C)$ whenever 
  $s\mathrel{R}t$ and $C \in \Sigma(R)$. We say 
  that  $s$  and  $t$ 
  are \emph{state bisimilar}, denoted by  $s\sim_s t$, if there exists
  some state bisimulation  $R$ such that $s\mathrel{R}t$. The 
  relation $\sim_s$ is called \emph{state bisimilarity}.
\end{definition}

\begin{definition}
  Let $\lmp{S}=(S,\Sigma,\{\tau_a \mid a\in L\})$ an LMP and
  $\Lambda\subseteq\Sigma$.  $\Lambda$ is \emph{stable} with respect
  to $\lmp{S}$ if for all $A \in \Lambda$, $r \in  [0,1]\cap \Q$ and
  $a \in L$, we have $\{s : \tau_a(s,A)>r\} \in \Lambda$.
\end{definition}

Note that for  a  sub-$\sigma$-algebra  $\Lambda\subseteq\Sigma$,
$\Lambda$ is  stable if and only if $(S,\Lambda,\{\tau_a\restriction 
S\times \Lambda \mid a\in
L\})$ is an LMP.

\begin{definition}
  Let $\lmp{S}=(S,\Sigma,\{\tau_a \mid a\in L\})$  be an LMP.
  A relation $R$ on $S$ will be called an \emph{event bisimulation}
  if  there exists a stable sub-$\sigma$-algebra 
  $\Lambda \subseteq \Sigma$ such that $R=\mathcal{R}(\Lambda)$.

  Two states $s$ and $t$ of an LMP are \emph{event bisimilar}, denoted by $s\sim_e t$, if
  there exists some event bisimulation  $R$ such that 
  $s\mathrel{R}t$. The relation $\sim_e$ is called \emph{event 
  bisimilarity}.
\end{definition}

To illustrate these concepts with the LMP $\lmp{U}$, one can show 
that $\Xi\defi \sigma(\Borel(I) \cup\{\{s,t\},\{x\}\})$ is a stable 
$\sigma$-algebra, hence $s$ and $t$ are event bisimilar. However 
they are not state bisimilar as $V$ is 
$\Rel(\sim_s)$-closed  and 
$\tau_{\infty}(s,V)\neq \tau_{\infty}(t,V)$. See 
\cite{Pedro20111048} for details.

Recall the modal logic  $\Logic$ presented in the Introduction.
We spell out the  formal interpretation of
the modalities: $s \models \langle a
\rangle_{>q}\phi$ on the LMP $(S,\Sigma,\{\tau_a \mid a\in L\})$ if 
and only
if there exists  $A \in \Sigma$ such that for all $s' \in A$, $s'
\models \phi$ and $\tau_a(s,A) > q$. Given a formula $\phi$ we write
$\llbracket \phi \rrbracket$ to denote the set of states satisfying 
$\phi$. It can be proved by induction that each of these sets is
measurable. We write $\llbracket \Logic \rrbracket$ for 
the collection of sets $\llbracket \phi \rrbracket$;  we have the
following logical characterization of event bisimilarity.

\begin{theorem}[{\cite[Prop.~5.5 and 
Cor.~5.6]{coco}}]\label{thm:sigma-logic-smallest-stable}
  For an LMP $(S,\Sigma,\{\tau_a \mid a\in L\})$, 
  $\sigma(\sem{\Logic})$ is
  the smallest stable $\sigma$-algebra included in
  $\Sigma$. Therefore the logic $\Logic$ characterizes event
  bisimilarity, in symbols 
  ${\sim_e}=\Rel(\sem{\Logic})=\Rel(\sigma(\sem{\Logic}))$.
\end{theorem}

The last equality follows easily from the fact that 
$\Rel(\sigma(\mathcal{F}))=\Rel(\mathcal{F})$ holds for any family of 
sets $\mathcal{F}$.

Regarding $\lmp{U}$, the $\sigma$-algebra $\Xi$ turns out to be 
$\sigma(\sem{\Logic})$, therefore 
${\sim_e}=\mathrm{id}_U\cup\{(s,t),(t,s)\}$. This further implies 
that 
${\sim_s}=\mathrm{id}_U$ given that it can be proved that 
${\sim_s}\sbq {\sim_e}$ is always the case and, as noted before, 
$s$ and $t$ are not state-bisimilar. Consequently, this is an 
example 
where state bisimilarity is properly contained in event 
bisimilarity. In fact, the LMP $\lmp{U}$ was introduced in 
\cite{Pedro20111048} to show that event bisimilarity and state 
bisimilarity differ in LMP over general measurable spaces. In this 
work it will serve as a seed for several constructions to be 
performed in Section~\ref{sec:zhou-ordinal}.

%% file: operators_O_G.tex
\section{The operators $\Op$ and $\G$}\label{sec:operators-op-g}

Fix a Markov process $\lmp{S}=(S,\Sigma,\{\tau_a\mid a\in L\})$. We 
will work 
with the operators
defined in 
\cite{DBLP:journals/entcs/Zhou13},  and we introduce a new one,
$\G$:

\begin{definition} 
  Let $\Lambda\subseteq\Sigma$ and $R \subseteq S\times S$,
  \begin{itemize}
  \item the relation $\Rel^T(\Lambda)$ is given by
    \[
    (s,t) \in \Rel^T(\Lambda) \iff \forall a\in L \;
    \forall E\in \Lambda \; \tau_a(s,E)=\tau_a(t,E).
    \]
  \item $\Op(R)\defi \Rel^T(\Sigma(R))$.
  \item $\G(\Lambda)\defi \Sigma(\Rel^T(\Lambda))$.
  \end{itemize}
\end{definition}

Note that  $\Op(R)$ is always an  equivalence relation for any 
$R$ and if $\Lambda$ is a $\sigma$-algebra, then
$\G(\Lambda)$ is too. 
The motivating idea behind the definition of $\Rel^T$ is to relate
states that are probabilistically indistinguishable with respect to a
fixed set of ``tests'', here given by the family $\Lambda$ of
events. An equivalence relation $R$ on the set of states induces
naturally  $\Sigma(R)$, the $R$-closed sets in $\Sigma$, as a family 
of tests. It follows that  $R$ is a state
bisimulation if and only if $R\sbq \Rel^T(\Sigma(R)) = \Op(R)$.

\begin{example}\label{exm:calculations-zhou-U}
  Let us do some simple calculations with the operator $\Op$ 
  concerning the LMP $\lmp{U}$. If $\nabla$ denotes the 
  total relation $U\times U$,
  \[
  \Op(\nabla)=\Rel^T(\Upsilon (\nabla))=\Rel^T(\{\emptyset,U\})=
  \mathrm{id}_U\cup\{(s,t),(t,s)\}={\sim_e}.
  \]
  Also, because $V \in \Upsilon(\sim_e)$, 
  \[
  \Op(\sim_e)=\Rel^T(\Upsilon(\sim_e))=\mathrm{id}_U={\sim_s}.
  \]
  We highlight the fact that a 
  single application of $\Op$ from the event bisimilarity $\sim_e$ 
  leads to state bisimilarity $\sim_s$.  
\end{example}

In the next proposition we collect  some basic facts on the operators
defined up to this point; most of them appear in 
\cite{DBLP:journals/entcs/Zhou13}.

\begin{proposition}\label{prop:basics-operator}
  Let $\Lambda, \Lambda'\subseteq \Sigma$ be sub-$\sigma$-algebras 
  and $R, R'\subseteq
  S\times S$.
  \begin{enumerate}
  \item\label{item:1} $\Lambda\subseteq \Sigma(\Rel(\Lambda))$.
  \item \label{item:2}$R\subseteq \Rel(\Sigma(R))$.
  \item\label{A3} If $\Lambda\subseteq \Lambda'$, then
    $\Rel(\Lambda)\supseteq\Rel(\Lambda')$ and
    $\Rel^T(\Lambda)\supseteq\Rel^T(\Lambda')$.
  \item\label{A4} If $R\subseteq R'$, then
    $\Sigma(R)\supseteq \Sigma(R')$.
  \item\label{A5}
    $\Rel(\Sigma(\Rel(\Lambda)))=\Rel(\Lambda)$
  \item \label{item:6}$\Op$ and $\G$ are monotone operators.
  \item \label{item:7} $R$ is a state bisimulation iff
    $(S,\Sigma(R),\{\tau_a\mid a \in L\})$ is an LMP.
  \item \label{item:8} If $\Lambda$ is stable then
    $\Rel(\Lambda)\subseteq \Rel^T(\Lambda)$.
  \end{enumerate}
\end{proposition}

We will also need some basic material on fixpoint theory. We work with
von Neumann ordinals, viz.\ $\alpha = \{\gamma : \gamma < \alpha\}$. If
$F:A\to A$ is a function on a complete lattice $A$, we
define the iterates of $F$ by $F^0(x)\defi x$, 
$F^{\alpha+1}(x)\defi F(F^\alpha(x))$,
$F^\lambda(x)\defi\bigwedge_{\alpha<\lambda}F^\alpha(x)$ if $\lambda$
is a limit ordinal, and
$F^\infty(x)=\bigwedge_{\lambda}F^\lambda(x)$. We say that $x$ is a
\emph{pre-fixpoint} (resp., \emph{post-fixpoint}) of $F$ if $F(x)\leq
x$ (resp., $x\leq F(x)$).
\begin{proposition}[{\cite[Exr.~2.8.10]{sangiorgi}}]
  \label{pre-punto fijo}
  Let $F:A\to A$ be a monotone function on a complete lattice 
  $A$. If $x$
  is a  pre-fixpoint of $F$, then $F^\infty(x)$ is the greatest
  fixpoint of  $F$ below $x$. Furthermore, this fixed point is 
  reached at an ordinal $\alpha$ such that $|\alpha|\leq|A|$.
\end{proposition}

As in Zhou's work \cite{DBLP:journals/entcs/Zhou13} we will construct
chains of relations  and of $\sigma$-algebras using the operators
$\Op$ and $\G$. 
The next result will be an aid in showing that $\sigma(\sem{\Logic})$ 
(respectively,
${\sim_e}=\mathcal{R}(\sigma(\sem{\Logic}))$) is a post(pre)-fixpoint 
of
$\G$ ($\Op$).

\begin{lemma}\label{r=rt}
  $\mathcal{R}^T(\sigma(\sem{\Logic}))=\mathcal{R}(\sigma(\sem{\Logic}))$.
\end{lemma}
\begin{proof}
  Since  $\sigma(\sem{\Logic})$ is stable, we have 
  $\Rel(\sigma(\sem{\Logic}))\subseteq\Rel^T(\sigma(\sem{\Logic}))$ 
  by Proposition~\ref{prop:basics-operator}(\ref{item:8}). We prove 
  the other
  inclusion by structural induction on formulas. Suppose that 
  $(s,t) \in  \Rel^T(\sigma(\sem{\Logic}))$. 
  If $A\defi \llbracket \top \rrbracket
  = S$ then $s \in A \Leftrightarrow t \in A$. 
  The case
  $A= \llbracket \phi \wedge \psi \rrbracket = \llbracket
  \phi \rrbracket \cap \llbracket \psi \rrbracket$ is also trivial 
  from the IH. For the case
  $A= \llbracket \langle a \rangle_{>q}\phi \rrbracket$, 
  observe that the hypothesis $(s,t) \in \Rel^T(\sigma(\sem{\Logic}))$
  implies $s\in A \Leftrightarrow \tau_a(s,\llbracket \phi
  \rrbracket)>q \Leftrightarrow\tau_a(t,\llbracket \phi
  \rrbracket)>q \Leftrightarrow t\in A$. Then, the
  $\sigma$-algebra $\mathcal{A}_{s,t}\defi \{A\in \Sigma\mid
  s\in A \Leftrightarrow t\in A\}$ includes $\llbracket \Logic
  \rrbracket$. We conclude that  $\sigma(\sem{\Logic})\subseteq 
  \mathcal{A}_{s,t}$,
  i.e., $(s,t)\in \Rel(\sigma(\sem{\Logic}))$.
\end{proof}

\begin{corollary}\label{cor:event-bisim-prefixpoint}
  $\sigma(\sem{\Logic})\subseteq \G(\sigma(\sem{\Logic}))$ and 
  $\Op(\sim_e)\subseteq {\sim_e}$.
\end{corollary}
\begin{proof}
  Since  $\Sigma\circ \Rel$ is expansive by
  Proposition~\ref{prop:basics-operator}(\ref{item:1}), the previous
  Lemma implies that  $\sigma(\sem{\Logic})\subseteq
  \Sigma(\Rel(\sigma(\sem{\Logic})))=
  \Sigma(\Rel^T(\sigma(\sem{\Logic})))=
  \mathcal{G}(\sigma(\sem{\Logic}))$.
   Moreover,
  by antimonotonicity of  $\Rel^T$ we obtain the result for 
   $\Op$:
  \[
  \Op(\sim_e)=\Rel^T(\Sigma(\Rel(\sigma(\sem{\Logic}))))\subseteq
  \Rel^T(\sigma(\sem{\Logic}))=\Rel(\sigma(\sem{\Logic}))={\sim_e}. 
  \qedhere
  \]
\end{proof}

The inclusions 
$\Op(R)\subseteq R$ and $\Lambda \subseteq\mathcal{G}(\Lambda)$ do not
hold in  general for arbitrary  $R$ and $\Lambda$. 
For example, if  $\tau$ is null for all arguments,  $\Op(R) =S\times S$ for any
relation $R$ and analogously for $\G$. 

Given a  relation  $R$ we define a  transfinite sequence of
equivalence relations  using the  operator $\Op$:
\begin{itemize}
\item $\Op^0(R)\defi R$;
\item
  $\Op^{\alpha+1}(R)\defi \Op(\Op^\alpha(R))$;
\item $\Op^{\lambda}(R)\defi \bigcap_{\alpha <
  \lambda}\Op^{\alpha}(R)$ if  $\lambda$ is  a
  limit.
\end{itemize}
Similarly, if $\Lambda \subseteq \Sigma$ is a  
$\sigma$-algebra,  $\G$ generates a family of
$\sigma$-algebras given by:
\begin{itemize}
\item $\G^0(\Lambda)\defi \Lambda$;
\item $\G^{\alpha+1}(\Lambda)\defi \G(\G^\alpha(\Lambda))$;
\item $\G^{\lambda}(\Lambda)\defi \sigma(\bigcup_{\alpha < \lambda}
  \G{^\alpha}(\Lambda))$ if $\lambda$ is a  limit ordinal.
\end{itemize}

Note that in the limit case of this last definition we must take the
generated $\sigma$-algebra since the union of a countable chain of 
$\sigma$-algebras is not in general a $\sigma$-algebra.

Let $\Sigma_0 \subseteq \Sigma$ be a  sub-$\sigma$-algebra and 
$R_0\subseteq S\times S$ a relation. From the iterates of 
$\Op$ and $\G$ we define  new $\sigma$-algebras and relations.
\begin{definition}
  For every  ordinal $\alpha$ let
  $\Sigma_\alpha\defi \mathcal{G}^\alpha(\Sigma_0)$ and
  $R_\alpha\defi \Op^\alpha(R_0)$.
\end{definition}

It is clear that if  $\alpha<\lambda$ then 
$R_\lambda\subseteq R_\alpha$ since 
$R_\lambda=\Op^\lambda(R_0)=\bigcap_{\beta <
  \lambda}\Op^\beta(R_0)=\bigcap_{\beta < \lambda}R_\beta\subseteq
R_\alpha$. 
It is also easy to verify from the definitions that
$\Sigma_\alpha \subseteq \Sigma_\lambda$. We are interested in
determining what other relationships hold among these
relations and $\sigma$-algebras. 
We are mainly concerned in the case in which  
$\Sigma_0=\sigma(\sem{\Logic})$ and
when $R_0$ is the relation of event bisimilarity; then by 
Lemma~\ref{r=rt} we have $R_0=\Rel^T(\Sigma_0)$.

\begin{proposition}\label{hip alfa}
  If $R_0=\Rel^T(\Sigma_0)$, then for all $\alpha$,
  $R_\alpha = \Rel^T(\Sigma_\alpha)$.
\end{proposition}
\begin{proof}
  By induction on  $\alpha$. The case $\alpha=0$ is included in the
  hypothesis. Now assume that it holds for
  $\alpha$. We calculate as follows
  $R_{\alpha+1}=\Op^{\alpha+1}(R_0)=\Op(\Op^\alpha(R_0))=\Op(R_\alpha)=\Rel^T(\Sigma(R_\alpha))=\Rel^T(\Sigma(\Rel^T(\Sigma_\alpha)))=\Rel^T(\G(\Sigma_\alpha))=\Rel^T(\G(\G^\alpha(\Sigma_0))=\Rel^T(\G^{\alpha+1}(\Sigma_0))=\Rel^T(\Sigma_{\alpha+1})$. Then
  it holds for  $\alpha+1$.
	
  Suppose now that the result holds for all
  $\alpha<\lambda$, with $\lambda$ a  limit ordinal. We have
  $R_\lambda=\Op^\lambda(R_0)=\bigcap_{\alpha<\lambda}\Op^\alpha(R_0)=\bigcap_{\alpha<\lambda}R_\alpha=\bigcap_{\alpha<\lambda}\Rel^T(\Sigma_\alpha)$. We
  will prove that the last  term equals
  $\Rel^T(\bigcup_{\alpha<\lambda}\Sigma_\alpha)$.  Let $s,t \in S$. 
  Then 
  \begin{align*}
    (s,t) \in
    \textstyle\bigcap_{\alpha<\lambda}\Rel^T(\Sigma_\alpha)
    &\Leftrightarrow  \forall \alpha<\lambda \, \forall a\in L\,\forall
    Q \ (Q\in\Sigma_\alpha \Rightarrow
    \tau_a(s,Q)=\tau_a(t,Q)) \\
    &\Leftrightarrow  \forall a\in L\, \forall Q
    \, \forall \alpha<\lambda \, (Q \in \Sigma_\alpha
    \Rightarrow \tau_a(s,Q)=\tau_a(t,Q))\\
    &\Leftrightarrow 
    \forall a\in L\,\forall Q \, (\exists \alpha<\lambda \, Q \in
    \Sigma_\alpha \Rightarrow \tau_a(s,Q)=\tau_a(t,Q))
    \\
    &\Leftrightarrow  \forall a\in L\,\forall Q \, (Q \in \textstyle
    \bigcup_{\alpha<\lambda}\Sigma_\alpha \Rightarrow
    \tau_a(s,Q)=\tau_a(t,Q)) \\
    &\Leftrightarrow  (s,t) \in
    \Rel^T(\textstyle
    \bigcup_{\alpha<\lambda}\Sigma_\alpha)
  \end{align*}
  \begin{claim*}
    $\Rel^T(\bigcup_{\alpha<\lambda}\Sigma_\alpha)=\Rel^T(\sigma(\bigcup_{\alpha<\lambda}\Sigma_\alpha))$
  \end{claim*}

  With this we can conclude since
  \[
  \textstyle\Rel^T(\sigma(\bigcup_{\alpha<\lambda}\Sigma_\alpha))=\Rel^T(\sigma(\bigcup_{\alpha<\lambda}\mathcal{G}^\alpha(\Sigma_0)))=\Rel^T(\mathcal{G}^\lambda(\Sigma_0))=\Rel^T(\Sigma_\lambda).
  \]
  Now we prove the claim. Let $s,t  \in S$ such that 
  $(s,t)\in \Rel^T(\bigcup_{\alpha<\lambda}\Sigma_\alpha)$. 
  We define
  $\mathcal{D}_{s,t}\defi\{A\in \Sigma \mid  \forall a\in L, \, 
  \tau_a(s,A)=\tau_a(t,A)\}$.
  We check that $\mathcal{D}_{s,t}$ is a monotone class on
  $S$. If $\{A_i\}_{i\in\omega}$ is an increasing 
  family of subsets  $S$ such that $A_i\in
  \mathcal{D}_{s,t}$, upper continuity of the  measures
  $\tau_a(s,\cdot)$ and $\tau_a(t,\cdot)$ implies
  \[ 
  \tau_a(s,\textstyle \bigcup_{i\in\omega}A_i)=\lim_i\tau_a(s,A_i)=\lim_i\tau_a(t,A_i)=\tau_a(t,\textstyle \bigcup_{i\in\omega}A_i)
  \]
  We argue similarly for an intersection of a decreasing family 
  in $\mathcal{D}_{s,t}$ by using lower continuity of the (finite)
  measures involved.  Since, by hypothesis,
  $\bigcup_{\alpha<\lambda}\Sigma_\alpha \subseteq
  \mathcal{D}_{s,t}$ and moreover, the family
  $\bigcup_{\alpha<\lambda}\Sigma_\alpha$ is an algebra of subsets of
  $S$, the Monotone Class Theorem yields 
  $\sigma(\bigcup_{\alpha<\lambda}\Sigma_\alpha)\subseteq\mathcal{D}_{s,t}$.
  
  Since the reverse inclusion 
  $\Rel^T(\bigcup_{\alpha<\lambda}\Sigma_\alpha)\supseteq\Rel^T(\sigma(\bigcup_{\alpha<\lambda}\Sigma_\alpha))$
  is trivial, we have the result.
\end{proof}

\begin{corollary}\label{cor:Sigma-R-alpha}
  If $R_0=\Rel^T(\Sigma_0)$, then for all $\alpha$,
  $\Sigma(R_\alpha)=\Sigma_{\alpha+1}$.
\end{corollary}
\begin{proof}
  Unfolding definitions,
  \[
  \Sigma(R_\alpha)=
  \Sigma(\mathcal{R}^T(\Sigma_\alpha))=\Sigma(\mathcal{R}^T(\mathcal{G}^\alpha(\Sigma_0)))=\mathcal{G}(\mathcal{G}^{\alpha}(\Sigma_0))=\mathcal{G}^{\alpha+1}(\Sigma_0)=\Sigma_{\alpha+1}. \qedhere
  \]
\end{proof}

\begin{proposition}\label{prop:R-decreasing}
  If $R_0=\mathcal{R}(\Sigma_0)=\mathcal{R}^T(\Sigma_0)$,
  then for all $\alpha$, $R_{\alpha+1}\subseteq  R_\alpha$.
\end{proposition}
\begin{proof}
  We work by induction on $\alpha$. By using the antimonotonicity of 
  $\mathcal{R}^T$ and 
  $\Sigma(\mathcal{R}(\Sigma_0))\supseteq\Sigma_0$ we have 
  \[
  R_1=\Op(R_0)=\mathcal{R}^T(\Sigma(R_0))=\mathcal{R}^T(\Sigma(\mathcal{R}(\Sigma_0)))\subseteq
  \mathcal{R}^T(\Sigma_0)=R_0.
  \]
  This shows the base case.
  Assume the result for  $\alpha$, then by applying the  monotonicity
  of  $\Op$  to the IH we have
  $R_{\alpha+2}=\Op(R_{\alpha+1})\subseteq
  \Op(R_\alpha)=R_{\alpha+1}$.  

  For limit  $\lambda$ we observe that  for all
  $\alpha<\lambda$ monotonicity of  $\Op$ and IH
  ensure 
  $R_{\lambda+1}=\Op(R_\lambda)\subseteq
  \Op(R_\alpha)=R_{\alpha+1}\subseteq R_\alpha$. Then,
  $R_{\lambda+1}\subseteq \bigcap_{\alpha < \lambda}R_\alpha=R_\lambda$.
\end{proof}

\begin{corollary}\label{cor:Sigma-monotonous}
  If $R_0=\mathcal{R}(\Sigma_0)=\mathcal{R}^T(\Sigma_0)$,
  then for all $\alpha$,
  $\Sigma_\alpha\subseteq\Sigma_{\alpha+1}$.
\end{corollary}
\begin{proof}
  For $\alpha=0$ we observe that 
  $\Sigma_0\subseteq\Sigma(\mathcal{R}(\Sigma_0))=\Sigma(\mathcal{R}^T(\Sigma_0))=\mathcal{G}(\Sigma_0)=\Sigma_1$.
  For successor ordinals, we use Proposition~\ref{prop:R-decreasing} 
  and Proposition~\ref{hip alfa} to obtain:
  \[
  \Sigma_{\alpha+1}=\mathcal{G}(\Sigma_\alpha)=\Sigma(\mathcal{R}^T(\Sigma_\alpha))=\Sigma(R_\alpha)\subseteq\Sigma(R_{\alpha+1})=\Sigma_{\alpha+2}.
  \]
  Finally, for the  limit case, observe that for all
  $\alpha<\lambda$, $\Sigma_\alpha\subseteq \Sigma_\lambda$; then by
  IH and monotonicity of  $\mathcal{G}$
  we have $\Sigma_\alpha\subseteq
  \Sigma_{\alpha+1}=\mathcal{G}(\Sigma_\alpha)\subseteq
  \mathcal{G}(\Sigma_\lambda)=\Sigma_{\lambda+1}$. Therefore
  $\Sigma_\lambda=\sigma(\bigcup_{\alpha <
    \lambda}\Sigma_\alpha)\subseteq\Sigma_{\lambda+1}$.
\end{proof}

In the following, for any $Q\in \Sigma$, $\posleq{a}Q$ will denote 
the set $\{s\in S 
  \mid \tau_a(s,Q)\leq q \}$. Similarly for the other order relations.
   
\begin{note}\label{note:RT-is-R}
  If $\Theta$ is a $\sigma$-algebra, a relation $\Rel^T(\Theta)$ is 
  of the form $\Rel(\mathcal{F})$ for some 
  subfamily $\mathcal{F}$ of $\Sigma$: If $\Gamma$ is any algebra such 
  that $\sigma(\Gamma)=\Theta$, then 
  \[
    s\mathrel{\mathcal{R}^T(\Theta)}t \iff 
    (s,t)\in \mathcal{R}(\{\posleq{a}Q \mid 
    a\in L,\, q\in\mathbb{Q},\, Q\in \Gamma\})
  \]
  Let $\mathcal{F}$ the family on the right hand side, namely 
  $\{\posleq{a}Q \mid a\in L,\, q\in\mathbb{Q},\, Q\in \Gamma\}$.
  If $(s,t)\in \mathcal{R}^T(\Theta)$ 
  then for any $a\in L$, $q\in \Q$
  and $Q\in \Theta$ we have
  $s\in \posleq{a}Q$ iff $\tau_a(s,Q)\leq q$ iff 
  $\tau_a(t,Q)\leq q$ iff $t\in
  \posleq{a}Q$. Conversely, suppose that 
  $(s,t)\in\mathcal{R}(\mathcal{F})$. Since $\mathcal{D}_{s,t}=\{Q\in
  \Theta\mid \forall a\in L, \, \tau_a(s,Q)=\tau_a(t,Q)\}$ is a 
  monotone class and 
  $\Gamma$
  is a generating algebra for  $\Theta$ such that $\Gamma\subseteq
  \mathcal{D}_{s,t}$, the Monotone Class Theorem yields 
  $\Theta=\sigma(\Gamma)\subseteq\mathcal{D}_{s,t}$ for every $a\in 
  L$.
\end{note}

\begin{proposition}\label{prop:RT-Sigma-lambda}
  If $R_0=\mathcal{R}(\Sigma_0)=\mathcal{R}^T(\Sigma_0)$,
  then for all limit ordinals  $\lambda$,
  $\mathcal{R}(\Sigma_\lambda)=\mathcal{R}^T(\Sigma_\lambda)$.
\end{proposition}
\begin{proof}
  By Note~\ref{note:RT-is-R} there is some $\Lambda\subseteq\Sigma$ 
  such that  $\Rel^T(\Sigma_\alpha)=\Rel(\Lambda)$. 
  From Proposition~\ref{prop:basics-operator}(\ref{A5})  it follows that
  \[
  \Rel^T(\Sigma_\alpha)=\Rel(\Sigma(\Rel^T(\Sigma_\alpha)))=\Rel(\G(\Sigma_\alpha))=\Rel(\Sigma_{\alpha+1})\subseteq
  \Rel(\Sigma_\alpha).
  \]
  If $\lambda$ is a  limit ordinal, 
  $\Sigma_{\alpha+1}\subseteq\Sigma_\lambda$ holds for all 
  $\alpha<\lambda$ and hence
  $\mathcal{R}(\Sigma_\lambda)\subseteq\mathcal{R}(\Sigma_{\alpha+1})=\mathcal{R}^T(\Sigma_\alpha)=R_\alpha$. Then,
  $\mathcal{R}(\Sigma_\lambda)\subseteq\bigcap_{\alpha<\lambda}R_\alpha=R_\lambda=\mathcal{R}^T(\Sigma_\lambda)=\mathcal{R}(\Sigma_{\lambda+1})\subseteq\mathcal{R}(\Sigma_\lambda)$
  and the result follows.
\end{proof}

In Section~\ref{sec:example} we will construct an LMP for which
$\Rel(\Sigma_{\alpha+1})\supsetneq\Rel^T(\Sigma_{\alpha+1})$,
and hence the previous equality does not hold for successor ordinals
in general.

In the case $R_0=\mathcal{R}(\Sigma_0)=\mathcal{R}^T(\Sigma_0)$, we 
may summarize the results up to this point in 
Figure~\ref{fig:chains_diagram}:
\begin{figure}[h]
  \begin{center}          
    $\xymatrix@R=10ex@C=10ex{ 
      *+[l]{R_0 } \ar @{} [r] |*{\supseteq\cdots\supseteq} & *+[l]{
         R_\lambda} \ar@<.5ex>@{|->}[dr]^\Sigma
      \ar@{|->}[r]^{\Op} & R_{\lambda+1} \ar@<.5ex>@{|->}[dr]^\Sigma
      \ar@{|->}[r]^{\Op} & *+[r]{R_{\lambda+2} \supseteq
        \dots}\\
      *+[l]{\Sigma_0 } \ar @{} [r] |*{\subseteq\cdots\subseteq}
      \ar@{|->}[u]^{\mathcal{R}=\mathcal{R}^T} & *+[l]{
        \Sigma_\lambda} \ar@{|->}[r]_{\mathcal{G}}
      \ar@{|->}[u]^{\mathcal{R}=\mathcal{R}^T} & \Sigma_{\lambda+1}
      \ar@<.5ex>@{|->}[ul]^{\mathcal{R}} 
      \ar@<1ex>@{|->}[u]_{\mathcal{R}^T}
      \ar@{|->}[r]_{\mathcal{G}} & *+[r]{\Sigma_{\lambda+2}\subseteq 
      \dots}
      \ar@<0.5ex>@{|->}[ul]^{\mathcal{R}} 
      \ar@{|->}[u]_{\mathcal{R}^T} }$
  \end{center}
  \caption{\label{fig:chains_diagram}Chains of relations and 
  $\sigma$-algebras induced by $\Op$
    and $\G$ ($\lambda$ limit).}
\end{figure}

\begin{corollary}
  For a limit ordinal $\lambda$, if
  $\Sigma(\mathcal{R}(\Sigma_\lambda))=\Sigma_\lambda$ then $\mathcal{R}(\Sigma_\lambda)$ is a state bisimulation.
\end{corollary}
\begin{proof}
  $\Rel^T(\Sigma(\Rel(\Sigma_\lambda)))=\Rel^T(\Sigma_\lambda)=\Rel(\Sigma_\lambda)$. Therefore $\Rel(\Sigma_\lambda)$ is a state bisimulation. 
\end{proof}

We add an observation about Figure~\ref{fig:chains_diagram}: If
$\G(\Gamma)=\Gamma$ then
$\Op(\Rel^T(\Gamma))=\Rel^T\Sigma(\Rel^T(\Gamma))=\Rel^T(\G(\Gamma))=\Rel^T(\Gamma)$. This
means that a fixpoint in the lower part forces a fixpoint in the upper
part. By using the example in \cite{Pedro20111048} it can be seen that
the converse does not hold.

\begin{lemma}\label{lem:equivalences-Lambda-stable}
  Let  $\Lambda$ be a  sub-$\sigma$-algebra of $\Sigma$ such that
  $\Sigma(\mathcal{R}(\Lambda))=\Lambda$. The following are equivalent:
  \begin{enumerate}
  \item \label{item:3}$\Lambda$ is stable.
  \item \label{item:4}$\mathcal{R}(\Lambda) \subseteq \mathcal{R}^T(\Lambda)$.
  \item \label{item:5}$\mathcal{R}(\Lambda)$ is a state bisimulation.
  \end{enumerate}
\end{lemma}
\begin{proof}
  \ref{item:3} implies \ref{item:4} by Proposition~\ref{prop:basics-operator}(\ref{item:8}).	
  For \ref{item:4}$\Rightarrow$\ref{item:5}, observe that 
  $\Rel(\Lambda)\sbq\Rel^T(\Lambda)=\Rel^T(\Sigma(\Rel(\Lambda)))$
  and this means that   $\Rel(\Lambda)$ is a  state bisimulation.
	
  By virtue of Prop.~\ref{prop:basics-operator}, Item~\ref{item:7} 
  implies 
  $(S,\Sigma(\Rel(\Lambda)),\tau)$ is an LMP, but then
  $\Sigma(\Rel(\Lambda))=\Lambda$ is stable.
\end{proof}

\begin{example*}\label{exm:1}
  The hypothesis is necessary: On $[0,1]$, take
  $\Sigma=\Borel([0,1])$, $\Lambda$ to be the countable-cocountable 
  $\sigma$-algebra  and  $\tau(x,A)\defi \delta_x(A)$ for $x\in
  [0,\frac{1}{2}]$ and $\tau(x,A)\defi \frac{1}{2} \delta_x(A)$ if
  $x\in  (\frac{1}{2},1]$. Then
    $\Sigma(\Rel(\Lambda))=\Sigma\neq\Lambda$,
    $\Rel(\Lambda)=\mathrm{id}_{[0,1]}=\Rel^T(\Lambda)$
    and $\Lambda$ is not stable since, e.g., 
    $\{x     \mid 
    \tau(x,[0,1])>\frac{1}{2}\}=[0,\frac{1}{2}]\notin\Lambda$. 
    This
    example shows that \ref{item:4} does not imply \ref{item:3} in
    general. Since the identity relation is a state
    bisimulation, we also conclude that \ref{item:5} does not imply \ref{item:3} in
    general.
\end{example*}
The  $\sigma$-algebra $\Sigma_\omega$ corresponding to the  LMP to be
presented in Section~\ref{sec:example}, satisfies
$\Rel^T(\Sigma_\omega)=\Rel(\Sigma_\omega)$ but $\Rel(\Sigma_\omega)$
is not a state bisimulation. Hence \ref{item:4} does not imply
\ref{item:5} in general.  This example has a stable $\Sigma_0$
but $R_0=\Rel(\Sigma_0)$ is not state bisimulation; hence \ref{item:3}
does not imply \ref{item:5}.

We now aim to prove that $\Sigma_\lambda$ is stable for limit
$\lambda$. With the notation introduced before Note~\ref{note:RT-is-R}, we 
observe 
that $\Gamma$ is stable if and only if $\forall
a \in L \; \forall q \in \Q \; \forall Q\in \Gamma \; \langle
a \rangle_{\leq q}Q \in \Gamma$. 
Given a label $a$, we define the following set
\[
\mathcal{A}_a \defi \{Q\in \Sigma \mid \forall q \in \Q \,
(\langle a\rangle_{\leq q}Q \in \Sigma_\lambda \, \wedge \,
\langle a \rangle_{<q}Q \in \Sigma_\lambda)\}.
\]
Then, to show that   $\Sigma_\lambda$ is stable it is enough to prove
$\Sigma_\lambda \subseteq \A_a$ for all $a\in L$.

\begin{lemma} \label{A_a}
  If $\lambda$ is a limit ordinal, then 
  $\forall a \in L \, \forall  \alpha<\lambda \; \Sigma_\alpha \subseteq  \mathcal{A}_a$.
\end{lemma}
\begin{proof}
  Let $\alpha <\lambda$ and $Q\in \Sigma_\alpha$. We will show that
  for every label $a\in L$ and for all $q\in \Q$ the sets 
  $\langle a\rangle_{\leq q}Q$ and $\langle  a\rangle_{<q}Q$ are in
  $\Sigma_{\alpha+1}=\Sigma(\mathcal{R}^T(\Sigma_\alpha))\subseteq \Sigma_\lambda$. 
  Since $\tau_a(\cdot,Q)$ is measurable,
  $\langle a\rangle_{\leq q}Q=\tau_a(\cdot, Q)^{-1}((0,q]) \in
    \Sigma$. 
    To check that  $\langle a\rangle_{\leq q}Q$ is
    $\mathcal{R}^T(\Sigma_\alpha)$-closed, note that 
    \begin{align*}
      s\mathrel{\mathcal{R}^T(\Sigma_\alpha)}t &\iff  \forall a\in L 
      \,
      \forall A\in \Sigma_\alpha \; \tau_a(s,A)=\tau_a(t,A)
      \\ 
      &\iff  \forall a\in L \, \forall A\in \Sigma_\alpha \,
      \forall q\in \Q \; (\tau_a(s,A)\leq q \text{ iff }
      \tau_a(t,A)\leq q) \\ 
      & \iff  \forall a\in L \, \forall
      A\in \Sigma_\alpha \, \forall q\in \Q \; (s\in 
      \langle a \rangle_{\leq q}A \text{ iff } t\in \langle 
      a \rangle_{\leq q}A)
    \end{align*}
    The proof for  $\langle a \rangle_{<q}Q$ is similar.
\end{proof}

\begin{lemma}\label{lem:A-properties}
  \begin{enumerate}
  \item\label{item:9} If $\{A_n\}_{n \in \omega}$ is a non-decreasing 
  sequence of measurable sets, then 
    for all $a\in L$ and for all $q\in \Q$, 
    $\langle   a \rangle_{\leq q}\bigcup_{n\in \omega}A_n=\bigcap_{n\in    \omega}\langle a \rangle_{\leq q}A_n$.
  \item If $\{B_n\}_{n\in \omega}$ is a non-increasing 
  sequence 
  of measurable sets, then 
    for all $a\in L$ and for all $q\in \Q$, $\langle
    a \rangle_{< q}\bigcap_{n\in \omega}B_n=\bigcup_{n\in    \omega}\langle a \rangle_{<q}B_n$.
  \item \label{item:A-monotonous}$\mathcal{A}_a$ is a monotone class.
  \end{enumerate}
\end{lemma}
\begin{proof}
  \begin{enumerate}
  \item In general, if $A\subseteq B$ then $\langle  a \rangle_{\leq q}B\subseteq \langle a \rangle_{\leq q}A$ by monotonicity of measures. 
    Thus we have  ($\subseteq$). For ($\supseteq$) , if $s\in S$
    satisfies
    $\forall n\in \omega \, \tau_a(s,A_n)\leq q$, the continuity of 
    the
    measure $\tau_a(s,\cdot)$ yields
    $\tau_a(s,\bigcup_{n\in \omega}A_n)=\lim\tau_a(s,A_n)\leq q$.
		
  \item Similarly to \ref{item:9}, 
    \[
    \textstyle\bigcap_{n\in\omega}B_n\subseteq B_m \implies
    \langle a
    \rangle_{<q}\bigcap_{n\in\omega}B_n\supseteq \langle
    a \rangle_{<q}B_m \implies \langle a
    \rangle_{<q}\bigcap_{n\in\omega}B_n \supseteq
    \bigcup_{m\in\omega} \langle a
    \rangle_{<q}B_m.
    \]
    For the other  inclusion, if 
    $s\in    \langle a \rangle_{<q}\bigcap_{n\in\omega}B_n $,
    continuity of the measure $\tau_a(s,\cdot)$ implies
    $q>\tau_a(s,\bigcap_{n\in \omega}B_n)=\lim \tau(s,B_n)$.
    Then, there exists $n\in \omega$ such that
    $\tau_a(s,B_n)<q$ and hence $s\in \langle a \rangle_{<q}B_n$.
		
  \item Let  $\{A_n\}_{n \in \omega}\subseteq \mathcal{A}_a$ a 
  non-decreasing sequence of sets. 
    Let  $q\in  \Q$; part \ref{item:9} allows us to conclude 
    $\langle a \rangle_{\leq q}\bigcup_{n\in  \omega}A_n=\bigcap_{n\in \omega}\langle a
    \rangle_{\leq q}A_n\in \Sigma_\lambda$ and also 
    \[\textstyle
    \langle a \rangle _{<q}\bigcup_{n\in
      \omega}A_n=\bigcup_{m\in\omega}\langle a
    \rangle_{\leq q-1/m}(\bigcup_{n\in
      \omega}A_n)=\bigcup_{m\in\omega}\bigcap_{n\in
      \omega}(\langle a \rangle_{\leq q-1/m}A_n) \in
    \Sigma_\lambda.
    \]    
    Then, $\bigcup_{n\in \omega}A_n    \in \mathcal{A}_a$.
    
    Now let  $\{B_n\}_{n \in \omega}\subseteq \mathcal{A}_a$ be 
    non-increasing and let  $q\in \Q$. The second part yields
    $\langle a \rangle_{<q}\bigcap_{n\in
      \omega}B_n=\bigcup_{n\in \omega}\langle a
    \rangle_{<q}B_n\in \Sigma_\lambda$ and also
    \[\textstyle
    \langle a \rangle _{\leq q}\bigcap_{n\in
      \omega}B_n=\bigcap_{m\in\omega}\langle a
    \rangle_{<q+1/m}(\bigcap_{n\in
      \omega}B_n)=\bigcap_{m\in\omega}\bigcup_{n\in
      \omega}(\langle a \rangle_{<q+1/m}B_n) \in
    \Sigma_\lambda.
    \]    
    Then, $\bigcap_{n\in \omega}B_n \in \mathcal{A}_a$.\qedhere
  \end{enumerate}
\end{proof}

\begin{theorem}
  $\Sigma_\lambda$ is a
  stable $\sigma$-algebra for any limit ordinal $\lambda$.
\end{theorem}
\begin{proof}
  Since  $\bigcup_{\alpha<\lambda}\Sigma_\alpha$ is an algebra of 
  sets,
  Lemma~\ref{A_a},
  Lemma~\ref{lem:A-properties}(\ref{item:A-monotonous}) and the
  Monotone Class Theorem yield 
  $\Sigma_\lambda=\sigma(\bigcup_{\alpha<\lambda}\Sigma_\alpha)\subseteq\mathcal{A}_a$
  for any $a \in L$.
\end{proof}

%% file: zhou.tex
\section{The Zhou Ordinal}
\label{sec:zhou-ordinal}

Zhou expressed state bisimilarity as a fixpoint:

\begin{theorem}[{\cite[Thm.~3.4]{DBLP:journals/entcs/Zhou13}}]
  State bisimilarity $\sim_s$ is the greatest fixpoint of $\Op$.
\end{theorem}

By direct application of Proposition~\ref{pre-punto fijo} we get the 
following
\begin{theorem}\label{th:fixpoint_O}
  Let  $R$ be an equivalence  relation on  $S$ such that
  ${\sim_s}\subseteq  R$ and $\mathcal{O}(R)\subseteq R$, then there
  exists an ordinal
  $\alpha$ such that $|\alpha|\leq |S|$ and   
  $\mathcal{O}^\alpha(R)={\sim_s}$.
\end{theorem} 
\begin{corollary}[{\cite[Thm.~4.1]{DBLP:journals/entcs/Zhou13}}]
  State bisimilarity $\sim_s$ can be obtained by iterating
  $\mathcal{O}$ from the total relation or from event bisimilarity 
  $\sim_e$.
\end{corollary}
\begin{proof}
  Apply Theorem~\ref{th:fixpoint_O} and 
  Corollary~\ref{cor:event-bisim-prefixpoint}
\end{proof}

Thanks to this result we may define the following concept.
\begin{definition}
  The \emph{Zhou ordinal} of an LMP $\lmp{S}$, denoted
  $\Zh(\lmp{S})$, is the minimum $\alpha$ such that
  $\Op^\alpha(\sim_e)={\sim_s}$.  The Zhou ordinal of a class
  $\mathcal{A}$ of processes, denoted $\Zh(\mathcal{A})$, is 
  the supremum of the class
  $\{\Zh(\mathbb{S})\mid \mathbb{S}\in \mathcal{A}\}$ if it is bounded
  or $\infty$ otherwise.
\end{definition}

We will focus on the study of the Zhou ordinal of the class 
$\seplmp$ of separable
LMPs. It is immediate that it is bounded by the cardinal successor of 
$\card{\R}$.
\begin{lemma}
  $\Zh(\seplmp)\leq (2^{\aleph_0})^+$
\end{lemma}
\begin{proof}
  Every separable metrizable space $S$ satisfies $\card{S}\leq
  2^{\aleph_0}$, and hence the bound follows from
  Theorem~\ref{th:fixpoint_O}.
\end{proof}

Next we provide the last technical ingredient for the constructions to
be performed for our main Theorems~\ref{th:zhou-is-limit} and
\ref{th:cf-Zhou-gt-omega}. It a simple though essential
device to enlarge a given LMP in
such a way that the original structure is “isolated” and it does not
produce any side effect on the enlargement. 

Suppose that  $\lmp{T}=(T,\Sigma,\{\tau_a\mid a\in L\})$ is 
an LMP with label set $L$. 
Let $e\notin T$ be a new state and $L'=L\cup \{\oet, 
\lpet\}$ be an expansion of the label set by two new 
actions.
Over the measurable space  $(T^\ast,\Sigma^\ast)\defi 
(T\oplus\{e\},\Sigma\oplus\{\{e\},\emptyset\})$, we define 
a new LMP $\lmp{T}^\ast$ with kernels $\{\tau^\ast_a\mid 
a\in L\} \cup \{\tau_\oet,\tau_\lpet\}$ given by 
\begin{equation}\label{eq:T-ast}
  \begin{split}
    \tau^\ast_a(r,Q)&\defi \chi_T(r)\cdot \tau_a(r,Q\cap T) \\
    \tau_\oet(r,Q)&\defi \chi_T(r)\cdot \delta_e(Q)\\
    \tau_\lpet(r,Q)&\defi \chi_{\{e\}}(r)\cdot \delta_e(Q).
  \end{split}
\end{equation}
It is clear that  $\lmp{T}^\ast$ is an LMP. The kernel $\tau_\oet$ 
allows, with probability $1$, a transition to  $e$ from each state
$t\in T$, and $\tau_\lpet$ only loops around $e$.

The use of a new state and two extra kernels (instead of just a single
new kernel) stems from the fact that in this way it is immediate that 
$\Rel^T$ (as a set operator) is the same, modulo $e$, for $\lmp{T}$ and
$\lmp{T}^\ast$. This has the following consequence, which will be used
in the sequel.
\begin{lemma}\label{lem:T-ast-chain}
  The Zhou ordinal is invariant under the map
  $\lmp{T}\mapsto\lmp{T}^\ast$, namely: $\Zh(\lmp{T}^\ast) =\Zh(\lmp{T})$.
\end{lemma}

The $\lmp{T}\mapsto\lmp{T}^\ast$ construction will be used in
conjunction with the next lemma, where $\bar{S}$ is the intended
enlargement that we referred to above.
\begin{lemma}\label{lem:zigzag-zhou-chain}
  Let $\lmp{S}=(S,\Sigma,\{\tau_a\mid a\in L\})$ be an LMP and 
  $\lmp{S'}=(S',\Sigma',\{\tau'_a\mid a\in L\})$ be an LMP 
  over a direct sum $(S',\Sigma')=(S\oplus 
  \bar{S},\Sigma\oplus \bar{\Sigma})$
  such that:
  \begin{itemize}
  \item
    for all $r\in S$ and 
  $a\in L$, $\tau'_a(r,Q)=\tau_a(r,Q\cap   S)$;
  \item
    $\Sigma'_0  =\sigma(\sem{\Logic}_{\lmp{S}'})$ and
    $\Sigma_0   =\sigma(\sem{\Logic}_\lmp{S})$; and 
  \item
    $S\in \Sigma'_0$.
  \end{itemize}
  Then $\Sigma'_\alpha\rest S 
  =\Sigma_\alpha \sbq 
  \Sigma'_\alpha$
  (equivalently,
  $\Sigma'_\alpha = \Sigma_\alpha\oplus 
  \Sigma'_\alpha\rest 
  \bar{S}$) 
  and  $R'_\alpha = R_\alpha \cup 
  R'_\alpha\rest\bar{S}$ hold for every
  $\alpha\geq 0$.\qed
\end{lemma}
Regarding the equation that involves $R'_\alpha$, it says that
to know such relation it is enough to determine it in each direct
summand separately.
One interpretation of this is that whenever $S\in \Sigma'_0$, no
relevant information about $\lmp{S}$ is lost in the direct sum.

The proofs of the main results of this section are based on further analysis
of the LMP $\lmp{U}$ from Example~\ref{exm:lmp-U}, which  
was the first example of a process with positive Zhou ordinal. 
Actually, $\Zh(\lmp{U})=1$ 
as highlighted in Example~\ref{exm:calculations-zhou-U}.

A key idea behind the definition of $\lmp{U}$ is that the
non-measurable set $V$ is  essentially the only set that distinguishes
$\leb_0$ from $\leb_1$ and hence $s$ from $t$. This $V$ can become 
``available'' when all the rational intervals can be used to
separate points in $\Inter=(0,1)$. From this approach one can control 
the 
unveiling of $V$ using $B_n=(0,q_n)$ to become ``available 
as tests'' 
simultaneously or \textit{in parallel}, and this is the reason
why state bisimilarity is reached in one step in $\lmp{U}$. The same 
pattern will be used in Theorem~\ref{th:zhou-is-limit}. On the 
other hand a \textit{serial} approach to the uncovering of the family 
$\{B_n\}$ will be followed in the proof of 
Theorem~\ref{th:cf-Zhou-gt-omega}.

We will prove our first important result about $\Zh(\seplmp)$,
namely, that it is a limit ordinal. In order to do this we first give 
the construction of an LMP that will play an essential role in the 
aforementioned result. Since we will be exclusively concerned with the
Zhou ordinal from now on, $\Sigma_0$ will always be the least stable
$\sigma$-algebra $\sigma(\sem{\Logic})$ of the LMP in consideration
and $R_0 \defi \Rel(\Sigma_0)$.
Start with any $\lmp{T}$ such that 
$\Zh(\lmp{T})=\alpha+1$. 
Consider the  LMP $\lmp{T}^\ast$ constructed in 
(\ref{eq:T-ast}), that for simplicity we will denote by 
$\lmp{S}=(S,\Sigma,\{\tau_m\mid m\in \omega\}\cup 
  \{\tau_\oet,\tau_\lpet\})$.    
By Lemma~\ref{lem:T-ast-chain}, 
$\Zh(\lmp{S})=\alpha+1$. Let $z,y \in S\sm\{e\}$ be such 
that $z \mathrel{R_\alpha} y$ but $z 
\mathrel{{\cancel{R_{\alpha+1}}}} y$.
Then there exist $A_0\in  \Sigma_{\alpha+1}\setminus \Sigma_\alpha$ 
and $n\in \omega$ such that
$\tau_n(z,A_0)\neq \tau_n(y,A_0)$. We now define a new  process: Let 
\[
\lmp{S}'=\bigl(S \oplus \Inter \oplus\{s,t\},\ 
\Sigma \oplus\Borel_V\oplus\Power(\{s,t\}),\ 
\{\tau'_m\}_{m\in 
  \omega}\cup\{\tau'_\oet,\tau'_\lpet,\tau'_\infty\} \bigr)
\]
where
\begin{align*}
  \tau'_m(r,Q) &\defi
  \begin{cases*}
    \tau_m(r,Q\cap S) & if $r\in S$ \\ 
    \tau_n(z,Q\cap S) & if $r\in \Inter , r<q_m$ \\ 
    \tau_n(y,Q\cap S) & if $r\in \Inter , r\geq q_m$ \\ 
    0 & if $r\in \{s,t\}$
  \end{cases*} \\
  \tau'_\Box(r,Q) &\defi \chi_S(r)\cdot \tau_\Box(r,Q\cap 
  S) \quad \text{ for } \Box \in \{\oet,\lpet\}\\
  \tau'_{\infty}(r,Q) &\defi \chi_{\{s\}}(r)\cdot 
  \leb_0(Q\cap \Inter) +\chi_{\{t\}}(r)\cdot \leb_1(Q\cap \Inter ).
\end{align*}

We will call $\Sigma'$ the $\sigma$-algebra of $\lmp{S}'$ and 
anything referred to this LMP will be primed:
$\Sigma'_\alpha, 
R'_\alpha, \sim'_s$.

This new process will act as an amalgam of $\lmp{S}$ and 
$\lmp{U}$ with $x$ replaced by $S$:
Each state in $I$ behaves either
as $z$ or $y$
according to the label $m\in\omega$, 
and the process continues inside   $\lmp{S}$
afterwards.
Labels $\oet$ and $\lpet$ allow to separate the LMP 
$\lmp{S}$  from the  
rest in such a way that its behavior is independent of the enlargment.
If that were not the case, event bisimilarity could identify states of
$S$ and $I\cup \{s,t\}$, and therefore restrict the sets 
that appear in $\Sigma'_\alpha\rest S$.
Observe that $\lmp{S}'$ will end up with infinitely many different
kernels, even though 
$\lmp{S}$ had only finitely many.
Also note that for  $r\in I$, there are only three possible values of  $\tau'(r,Q)$: $\tau_n(z,Q\cap S)$, 
$\tau_n(y,Q\cap S)$ or $0$; this is very similar to 
$\lmp{U}$, where there were only two possible values of 
$\tau_n(r,Q)$.  

\begin{lemma}\label{lem:S'-LMP-and-restriction-sigma-prime}
  $\lmp{S}'$ is an LMP. Moreover, $\forall \beta \; 
  \Sigma'_\beta\rest 
  S = \Sigma_\beta \sbq \Sigma'_\beta$ and $R'_\beta = 
  R_\beta \cup R'_\beta\rest I\cup \{s,t\}$.
\end{lemma}
\begin{proof}
  To show that  $\lmp{S}'$ is an LMP, we only need to check that 
  $\tau'_a$ is a  Markov kernel for every
  $a\in \omega\cup\{\infty,\oet,\lpet\}$. 

  If $Q\in \Sigma'$, measurability of $\tau'_m(\cdot,Q)$ 
  follows from the fact that $\tau_m(\cdot,Q\cap S)$ is measurable 
  for all $m\in \omega$ and from the measurability of the sets  $(0,q_m)$ and $\{s,t\}$.
  Measurability of $\tau'_\Box(\cdot,Q)$ for $\Box\in 
  \{\oet,\lpet,\infty\}$ only depends on the measurability of
  the characteristic functions involved.  
  Finally, for fixed $r\in S'$, all maps 
  $\tau'_a(r,\cdot)$ are clearly
  subprobability measures.

  For the second statement, consider the LMP obtained by adding the   zero kernel with $\infty$ label to 
  $\lmp{S}$.
  This operation does not modify     
   $R_\alpha$ nor $\Sigma_\alpha$. 
  Moreover, it is immediate that for all $r\in S$ and 
  labels $a$, $\tau'_a(r,Q)=\tau_a(r,Q\cap 
  S)$ holds. %
  Note that also $S\in \Sigma_0'$, since 
  $\sem{\pos{\lpet}_{>0}\top}_{\lmp{S}'} \cup 
  \sem{\pos{\oet}_{>0}\top}_{\lmp{S}'} =S$. In this way, we may apply
  Lemma~\ref{lem:zigzag-zhou-chain} to 
  $\lmp{S}$ and the measurable space $(\Inter 
  \oplus\{s,t\},\Borel_V\oplus \Power(\{s,t\}))$ to obtain 
  the result.
\end{proof}

We are now ready to prove the previously announced result. 
\begin{theorem}\label{th:zhou-is-limit}
  $\Zh(\seplmp)$ is a limit  ordinal.
\end{theorem}
\begin{proof}
  First observe that $\Zh(\seplmp)>0$ as shown in 
  \cite{Pedro20111048}. Suppose by way of contradiction that 
  $\Zh(\seplmp)=\alpha+1$ for some $\alpha\geq 0$. Then there must 
  exist a separable LMP $\lmp{T}$ such that 
  $\Zh(\lmp{T})=\alpha+1$. 
  Now 
  consider the LMP $\lmp{S'}$ as in the previous construction. We 
  show that 
  $\Zh(\lmp{S}')\geq \alpha+2$. 
  To see this it is enough to prove that  $s\mathrel{R'_{\alpha+1}}t$ 
  but
  $s\mathrel{\cancel{R'_{\alpha+2}}}t$. For the first condition,
  let us show that $\Sigma'_{\alpha+1}\restriction \Inter =\{\emptyset, 
  \Inter \}$.
  Let  $Q\in  
  \Sigma'_{\alpha+1}=\Sigma'(\Rel^T(\Sigma'_\alpha))$ and 
  assume
  $Q\cap\Inter \neq \emptyset$; we show that $Q\cap \Inter =\Inter $. 
  Let  $r_0\in Q\cap \Inter $ and $r\in \Inter $. Suppose
  that  $r_0 \mathrel{\cancel{\Rel^T(\Sigma'_\alpha)}}r$; then
  there exist $m\in \omega$ and $B\in \Sigma'_\alpha$ such that
  $\tau'_m(r_0,B)\neq \tau'_m(r,B)$, i.e.\ 
  $\tau_m(z,B\cap S)\neq \tau_m(y,B\cap S)$. By
  Lemma~\ref{lem:S'-LMP-and-restriction-sigma-prime}, 
  $B\cap S\in  \Sigma_\alpha$, then 
  $z    \mathrel{\cancel{\Rel^T(\Sigma_\alpha)}} y$; but 
  this is absurd
  since we chose $z,y$ in such a way they are indeed related. It
  follows that 
  $r_0\mathrel{\Rel^T(\Sigma'_\alpha)} r$ and since $Q$ is
  $\Rel^T(\Sigma'_\alpha)$-closed, $r\in Q\cap \Inter $. Note that 
  this yields $R'_\alpha\restriction \Inter =\Inter \times \Inter $.
  To show  $(s,t) \in    R'_{\alpha+1}=\Rel^T(\Sigma'_{\alpha+1})$, 
  consider
  $\emptyset\neq Q\in \Sigma'_{\alpha+1}$. By the previous
  calculation, 
  $Q\cap  \Inter =\Inter $, therefore
  $\tau'_\infty(s,Q)=1=\tau'_\infty(t,Q)$. 
  For the remaining labels $a\in \omega\cup\{\oet,\lpet\}$ 
  we have $\tau'_a(s,Q)=0=\tau'_a(t,Q)$.

  We now show that 
  $s\mathrel{\cancel{R'_{\alpha+2}}}t$. Recall that we had chosen  
  $z,y$ and $A_0\in \Sigma_{\alpha+1}\setminus \Sigma_\alpha$
  such that $\tau_n(z,A_0)\neq \tau_n(y,A_0)$ for some 
  $n\in\omega$.
  By Lemma~\ref{lem:S'-LMP-and-restriction-sigma-prime},
  $A_0\in\Sigma'_{\alpha+1}$ and from this we
  conclude
  $R'_{\alpha+1}\restriction
  \Inter =\Rel^T(\Sigma'_{\alpha+1})\restriction
  \Inter =\mathrm{id}_{\Inter }$. 
  We also observe that $I=S'\sm (S\cup 
  \sem{\pos{\infty}_{>0}\top})\in \Sigma'_0$, therefore
  $\Sigma'_{\alpha+2}\restriction \Inter =\{A\sbq I\mid 
  A\in \Sigma'(R'_{\alpha+1})\} = 
  \Sigma'\restriction\Inter =\Borel_V$. 
  Then we have $V\in \Sigma'_{\alpha+2}$ and
  using that set with the transition labelled by $\infty$ 
  we obtain $s\mathrel{\cancel{R'_{\alpha+2}}}t$.
\end{proof}

It can be deduced from the proof of the previous 
Theorem that from every
separable process with Zhou ordinal $\alpha+1$ another one can be
constructed with ordinal equal to $\alpha+2$. In spite of 
this, 
this construction does not allow to construct a process with
positive Zhou ordinal from one having null Zhou ordinal
(i.e., having the Hennessy-Milner property).

%% file: cofinality.tex
In the next theorem, the \emph{cofinality} $\cf(\lambda)$ of a limit
ordinal $\lambda$ is the least order type (equivalently, the least
cardinal) of an unbounded subset of $\lambda$.
\begin{theorem}\label{th:cf-Zhou-gt-omega}
  $\cf(\Zh(\mathcal{S}))>\omega$.
\end{theorem}
\begin{proof}
  Towards a contradiction, suppose that for every $m\in\omega$ we have 
  a separable $\lmp{S}_m=(S^m,\Sigma^m,\{\tau_n^m \}_{n\in \omega})$
  with label set $\{(m,n)\mid n\in \omega\}$ 
  such that $\zeta_m\defi\Zh(\lmp{S}_m)$  satisfy
  $\sup_{m\in\omega}\zeta_m=\Zh(\seplmp)$.  
  We will assume that these LMP have gone through the 
  construction given in \eqref{eq:T-ast}; this way each 
  process now has two distinguished labels, which for ease of
  reference we call $(m,\oet)$ and $(m,\lpet)$, that allow them to be 
  differentiated from each other with formulas.
  
  We can assume $\zeta_0>0$ and also that 
  $\{\zeta_m\}_{m\in \omega}$ 
  is a strictly increasing sequence; for convenience, we 
  set $\zeta_{-1}\defi 0$. 
  In this way  $\zeta_{m-1}<\zeta_m$  
  for all $m\geq0$, hence we can choose 
  $x_m, y_m \in S^m$ such that $x_m
  \mathrel{R^m_{\zeta_{m-1}}} y_m$ but $x_m
  \mathrel{\cancel{R^m_{\zeta_m}}} y_m$. Then there is a set $A_m \in 
  \Sigma_{\zeta_m}^m\minus \Sigma_{\zeta_{m-1}}^m$ such that for some $i\in \omega$ we have 
  $\tau_i^m(x_m,A_m)\neq\tau_i^m(y_m,A_m)$.
  By reordering the labels of the Markov kernels, we can assume that 
  $i\in \omega$ is exactly $m$.
  
  Let us define a new LMP with label set $L\defi\{(m,n)\mid
  m,n\in\omega\}\cup\{\infty\}$:
  \[\textstyle
  \mathbb{S}\defi \bigl( \bigl(\bigoplus_{m \in \omega}S^m\bigr)
  \oplus \Inter \oplus\{s,t\},\ 
  \bigl(\bigoplus_{m\in\omega}\Sigma^m\bigr)\oplus\Borel_V\oplus\Power
  (\{s,t\}),\ 
  \{\tilde{\tau}_n^m\}_{m,n\in
    \omega}\cup\{\tilde{\tau}_\infty\} \bigr).
  \]
  where the kernels are given by
  \begin{equation*}
    \tilde{\tau}^m_n(r,Q) =
    \begin{cases*}
      \tau^m_n(r,Q\cap S^m) & if $r\in S^m$ \\
      \tau^m_m(x_m,Q\cap S^m) & 
      if $r\in (0,q_m)$ and $m=n$ \\
      \tau^m_m(y_m,Q\cap S^m) & 
      if  $r\in [q_m,1)$ and $m=n$ \\
        0 & otherwise,
    \end{cases*}
  \end{equation*}
  \[
  \tilde{\tau}_{\infty}(r,Q) = \chi_{\{s\}}(r)\cdot 
  \leb_0(Q\cap\Inter )+\chi_{\{t\}}(r)\cdot \leb_1(Q\cap\Inter).
  \]
  
  In this case the LMP $\lmp{S}$ is an amalgam of the sum of all of 
  the $\lmp{S}^m$ and 
  $\lmp{U}$. The sets $(0,q_n) = B_n$ will become available 
  successively, using a serial approach to uncovering the
  non-measurable set $V$. In this way we can surpass the limit of the 
  Zhou ordinals of the $\lmp{S}^m$. 
  
  We will call $(S,\Sigma)$ the measurable space of 
  $\lmp{S}$. 
  It is easy to see that this indeed defines an LMP and the 
  separability of the base space follows from the 
  separability of each of the countably many summands that make 
  it up.
  We will show that $\Zh(\mathbb{S})\geq \Zh(\mathcal{S})+1$, 
  reaching a contradiction. For this, it is enough to verify that  
  $(s,t)\in R_{\Zh(\seplmp)}\sm 
  R_{\Zh(\seplmp)+1}$. This will be implied by the facts 
  $\forall \eta\leq \Zh(\seplmp) \ \Sigma_\eta\rest 
  \Inter\sbq \Borel(\Inter)$ and $V \in 
  \Sigma_{\Zh(\seplmp)+1}\sm \Sigma_{\Zh(\seplmp)}$ which 
  in turn are a consequence of the equality 
  \begin{equation}\label{eq:cf-Sigma_alpha}
  \textstyle\Sigma_\eta\rest\bigl(\bigoplus_{m \in 
    \omega}S^m\bigr)\oplus 
  \Inter =\bigl(\bigoplus_{m \in \omega} 
  \Sigma_\eta^m\bigr)\oplus \Theta_\eta
  \end{equation}
  for all $\eta\leq \Zh(\seplmp)$, 
  where $\Theta_\eta$ is the $\sigma$-algebra on $\Inter $ 
  generated by the intervals $\{(0,q_m)\mid 
  \zeta_m<\eta\}$. 

  Before proving this, we notice that for each $m\in 
  \omega$, we can add to $\lmp{S}^m$ zero kernels 
  $\tau^m_\infty,\tau^j_n$ $(j\neq m)$ and get the property 
  $\forall r\in S \ \forall a\in L \ 
  \tilde{\tau}_a(r,Q)=\tau_a(r,Q\cap S)$, while not changing
  $\Sigma^m_\eta$ nor $R^m_\eta$. Also, thanks 
  to the distinguished labels $(m,\oet)$ and $(m,\lpet)$ (which 
  cannot correspond to the label $(m,m)$ in $\lmp{S}^m$), we have 
  $S^m=\sem{\langle (m,\lpet)\rangle_{>0} \top} \cup 
  \sem{\langle (m,\oet)\rangle_{>0} \top}\in \Sigma_0$. 
  This way, all the hypotheses of 
  Lemma~\ref{lem:zigzag-zhou-chain} are satisfied. 
  Then, for each $m\in \omega$ and $\eta$ we have 
  $\Sigma_\eta=\Sigma^m_\eta\oplus 
  \Sigma_\eta\rest(S\sm S^m)$ and $R_\eta = R^m_\eta 
  \cup R_\eta\rest(S\sm S^m)$. 
  Using the fact that there are countably many summands and also 
  that $\Inter,\{s,t\}\in \Sigma_0\sbq\Sigma_\eta$ 
  (because 
  $\{s,t\}=\sem{\langle\infty \rangle_{>0} \top}$), for all 
  $\eta$ we can conclude
  \begin{equation*}\label{eq:first-cf-Sigma_alpha}
    \textstyle \Sigma_\eta=\bigoplus_{m\in 
      \omega}\Sigma^m_\eta 
    \oplus     
    \Sigma_\eta\rest \Inter \oplus \Sigma_\eta\rest 
    \{s,t\}.
  \end{equation*}
  So all we have to do now is to show by induction on $\eta$ 
  that $\Sigma_\eta\rest\Inter =\Theta_\eta$. If 
  $\eta=0$ then obviously $\Theta_0=\{\emptyset,\Inter \}$ 
  so we have to show that $\Sigma_0\rest\Inter  =
  \sigma(\sem{\Logic})\rest\Inter $ is trivial. In 
  order to do this, it is enough to show that  $\{Q\in 
  \Sigma_0\mid 
  Q\cap\Inter \in\{\emptyset,\Inter \}\}$ is stable. Assume 
  that $Q\in \Sigma_0$ satisfies
  $ Q \cap\Inter \in\{\emptyset,\Inter \}$ and 
  $(\posg{a}Q) \cap\Inter \neq\emptyset$; hence $a=(m,n)$ 
  for 
  some $m,n\in\omega$ (since it is obvious from 
  the definition of $\tilde{\tau}_\infty$ that $a\neq \infty$). 
  Then there is 
  $r\in \Inter $ such that $\tilde{\tau}^m_n(r, Q )>q$. It 
  follows that $m=n$, otherwise the kernel would equal zero.
  
  From 
  $ Q \in  \Sigma_0$ we have $Q \cap S^m\in \Sigma^m_0\sbq
  \Sigma^m_{\zeta_{m-1}}$ and considering  $x_m \mathrel{R}_{\zeta_{m-1}} y_m$,
  we conclude that $\tilde{\tau}^m_m(\cdot, Q)$ is constant on
  $\Inter$. This yields $r'\in \posg{a}Q$ for every $r'\in \Inter$.
  This shows that $\{Q\in \Sigma_0\mid 
  Q\cap\Inter \in\{\emptyset,\Inter \}\}$ is stable. As this class is easily seen 
  to be a $\sigma$-algebra, we have $\sigma(\sem{\Logic}) = 
  \Sigma_0 \sbq\{Q\in \Sigma_0\mid 
  Q\cap\Inter \in\{\emptyset,\Inter \}\}$ and therefore 
  the result holds for $\eta=0$. 

  Assume now that $\Sigma_\eta\rest\Inter 
  =\Theta_\eta$. 
  Notice that the kernels in $\lmp{S}$ only depend on one, 
  and only one, of the 
  restrictions to $S^m$ and $\Inter $, and use this 
  together with the IH to obtain
  \[
  \textstyle \Rel^T(\Sigma_\eta) 
  =\Rel^T(\Sigma_\eta\rest \bigoplus_{m\in
    \omega}S^m \oplus\Inter)
  =\Rel^T(\bigoplus_{m\in\omega}\Sigma^m_\eta 
  \oplus\Theta_\eta)
  =\Rel^T\bigl(\bigoplus_{m\in\omega} 
  \Sigma^m_\eta\bigr)\cap\Rel^T(\Theta_\eta).
  \]
  As $\Rel^T(\Theta_\eta)= (S\sm\{s,t\}\times 
  S\sm\{s,t\})\cup \{(s,t),(t,s)\}$, 
  then $R_\eta$ is the disjoint union 
  $R_\eta=\bigl(\bigcup_{m\in\omega}R^m_\eta\bigr) \cup 
  (\Rel^T\bigl(\bigoplus_{m\in\omega}\Sigma_\eta^m\bigr)) 
  \rest \Inter \cup \{(s,t),(t,s)\}$.
  Therefore, $A\sbq \Inter$ is $R_\eta$-closed iff it is 
  $\Rel^T(\bigoplus_{m\in \omega}\Sigma^m_\eta)$-closed.  
  Also, from the choice of $x_m,y_m$ we deduce that  $\Rel^T(\bigoplus_{m\in\omega} \Sigma^m_\eta)\rest 
  \Inter =\Rel(\{(0,q_n)\mid \zeta_n\leq\eta\})\rest 
  \Inter$.
  From this we have
  \begin{align*}
    \Sigma_{\eta+1}\rest \Inter= \{A\in \Sigma_{\eta+1} 
    \mid A\sbq \Inter\} &= \textstyle 
    \{A\in \Sigma\mid A\sbq \Inter \wedge \text{is $ 
      \textstyle R_\eta$-closed}\} 
    \\ 
    &=\{A\in \Borel_V\mid A \text{ is 
      $\textstyle \Rel^T(\bigoplus_{m\in 
        \omega}\Sigma^m_\eta)$-closed}\} \\
    &= \sigma(\{(0,q_n)\mid \zeta_n\leq\eta\})   
    =\Theta_{\eta+1}.
  \end{align*} 
  
  For the limit case, assume 
  $\Sigma_\eta\rest\Inter =\Theta_\eta$ for all 
  $\eta<\lambda$. Then we have the following calculation 
  \[
  \Sigma_\lambda\rest\Inter =
  \textstyle\sigma\bigl( 
  \bigcup_{\eta<\lambda}\Sigma_\eta\bigr)\rest\Inter 
  =
  \textstyle\sigma\bigl( 
  \bigcup_{\eta<\lambda}(\Sigma_\eta\rest\Inter 
  )\bigr) =
  \textstyle\sigma\bigl( 
  \bigcup_{\eta<\lambda}\Theta_\eta\bigr)=\Theta_\lambda.
  \]
  This concludes  the proof by induction  of 
  Equation~\eqref{eq:cf-Sigma_alpha} and thus ends the proof 
  of the theorem.
\end{proof}
\begin{corollary}\label{cor:Zh-seplmp}
  $\Zh(\seplmp)\geq\omega_1$. \qed
\end{corollary}

%% file: example.tex
\section{The example}\label{sec:example}

In this section we construct, for each ordinal $\beta \leq \omega_1$,
an LMP $\lmp{S}(\beta)$ such that for $\beta$ limit, 
$\Zh(\lmp{S}(\beta))= \beta$. For this, we take the set $\Inter 
\times \beta$ together with the product
$\sigma$-algebra $\Sigma\defi \Borel_V \otimes \Power(\beta)$ where
$V$ denotes a Lebesgue non-measurable set.

Let $\{C_n\}_{n\geq 2}$ be the family of open rational intervals 
included
in $\Inter $ and set $C_0\defi V$ and  
$C_1\defi V^c$;  we have that
$\{C_n\}_{n\in \omega}$
generates $\Borel_V$. We now define a hierarchy of sets
$\boldsymbol{\Sigma}_\xi^{V}(I)$ totally analogous to the Borel
hierarchy.  $\boldsymbol{\Sigma}_1^{V}(I)$ is the family of (countable)
unions of sets in $\{C_n\}_{n\in
  \omega}$. The members of  $\boldsymbol{\Pi}_\xi^{V}(I)$ are the complements of sets
in $\boldsymbol{\Sigma}_\xi^{V}(I)$ and
$\boldsymbol{\Sigma}_{\xi}^{V}(I)\defi\bigl\{\bigcup_{n\in 
\omega}A_n\mid
A_n\in \boldsymbol{\Pi}_{\xi_n}^{V}(I), \xi_n<\xi\bigr\}$, for
$\xi>1$. Note that $\boldsymbol{\Sigma}_1^{V}(I)$ includes all the
open subsets of $I$ and their unions with $V$.

Given $Q\sbq I\times \beta$, the
\emph{sections} of $Q$ are the sets $Q_\alpha\defi\{r\mid
(r,\alpha)\in Q\}$ for $\alpha<\beta$ (to avoid confusion we will 
only use greek subindices for sections). Sometimes we will call 
$Q_\alpha$ ``$\alpha$-section'' to make the ordinal explicit. The 
sets  
$I\times \{\alpha\}$ will be
called \emph{fibers}. For  $Q$ in $\Borel_V \otimes \Power(\beta)$, each 
section  $Q_\alpha$  lies in 
$\boldsymbol{\Sigma}_\xi^{V}(I)$ for some $\xi$. We say that the
\emph{complexity} of  $Q$ at  $\alpha$ is  
$\comp(Q,\alpha) \defi\min \bigl\{\xi\mid
Q_\alpha\in \boldsymbol{\Sigma}_\xi^{V}(I)\bigr\}$, and the (total) complexity of
$Q$ is $\comp(Q) \defi \sup_{\alpha<\beta}\comp(Q,\alpha)$.
Sets in  $\Borel_V\otimes \Power(\beta)$ can be characterized in terms
of this complexity measure.

\begin{lemma}\label{lem:product_characterization}
  $Q\in \Borel_V\otimes \Power(\beta)$ if and only if $\comp(Q)<\omega_1$
\end{lemma}

\begin{proof}
  ($\Rightarrow$) Let us verify that  $\mathcal{A}=\{A\subseteq I\times \beta \mid \comp(A) < \omega_1\}$ is a 
  $\sigma$-algebra. %
  Assume $A \in \mathcal{A}$. Since
  $\comp(A^c,\alpha)\leq\comp(A,\alpha)+1$ then  
  $A^c \in  \mathcal{A}$.  Now  assume that 
  $\{A_n\}_{n\in \omega}\sbq\mathcal{A}$.
  Then   $\alpha_n \defi \comp(A_n)<\omega_1$ for all $n \in
  \omega$. 
  From this it follows that  $\comp(\bigcup_{n \in \omega}A_n)
  \leq \sup_{n \in \omega} (\alpha_n + 1) < \omega_1$, and therefore
  $\mathcal{A}$ is a $\sigma$-algebra.  Since 
  $\mathcal{A}$ includes all the measurable rectangles, we  have 
  $\Borel_V \otimes
  \Power(\beta)\subseteq \mathcal{A}$.
	
  ($\Leftarrow$) We show by induction on $\xi$ that $\comp(Q)=\xi
  <\omega_1 \implies Q\in \Borel_V\otimes \Power(\beta)$. If
  $\comp(Q)=1$ then for all $\alpha<\beta$, $Q_\alpha=\bigcup\{C_n\mid
  C_n\sbq Q_\alpha\}$ and therefore
  $Q=\bigcup_{n\in\omega}(C_n\times\{\alpha\mid C_n\sbq Q_\alpha\})\in
  \Borel_V\otimes\Power(\beta)$. Assume the result for all 
  $\eta$ with  $\eta<\xi<\omega_1$ and, moreover, assume
  $\comp(Q)=\xi$. Then
  $\forall\alpha<\beta$, $Q_\alpha\in \boldsymbol{\Sigma}_\xi^{V}(I)$.
  Hence
  $Q_\alpha=\bigcup_{n<\omega}A^\alpha_n$ for some 
  $A^\alpha_n \in \boldsymbol{\Pi}_{\xi_n(\alpha)}^V(I)\sbq
  \boldsymbol{\Sigma}_{\xi_n(\alpha)+1}^{V}(I)$ and 
  $\xi_n(\alpha)<\xi$ is non-decreasing.  Let 
  $\{\theta_n\}_{n\in\omega}$ such that $\theta_n+1$ is a 
  non-decreasing 
  sequence with limit $\xi$. If we set 
  $\tilde{A}^\alpha_n=\bigcup_{m\leq    n}\{A^\alpha_m\mid 
  A^\alpha_m\in
  \boldsymbol{\Sigma}^V_{\theta_n}(I)\}$, 
  then we have 
  $\forall  \alpha\; \tilde{A}^\alpha_n \in
  \boldsymbol{\Sigma}^V_{\theta_n}(I)$  
  and 
  $Q_\alpha=\bigcup_{n\in    \omega}\tilde{A}^\alpha_n$.
  By the IH, for every   $n\in \omega$, $\Borel_V \otimes  \Power(\beta)$ 
  contains the sets
  $C_n =  \bigcup_{\alpha<\beta}(\tilde{A}^\alpha_n\times\{\alpha\})$
  whose   $\alpha$-section is
  $\tilde{A}^\alpha_n$, with complexity 
  $\theta_n<\xi$. It follows that 
  $Q=\bigcup_{n\in \omega}C_n  \in \Borel_V \otimes \Power(\beta)$.
\end{proof}

We now define a denumerable family of Markov kernels. As before, fix 
an
enumeration $\{q_n\}_{n\in \omega}$ of the rational numbers in
$\Inter $. Define $\alpha_n(0)\defi 0$ ($n\in \omega$); and for each 
successor ordinal
 $\eta+1$, let $\alpha_n(\eta+1)\defi\eta$ ($n \in
\omega$). For limit $\lambda$ we choose $\{\alpha_n(\lambda)\}_{n\in
  \omega}$ to be a strictly increasing cofinal sequence in $\lambda -
\{0\}$ of order type $\omega$.

Recall that Theorem~\ref{thm:measure-extension} 
gives us two extensions $\leb_0$ and $\leb_1$ of the Lebesgue measure 
$\leb$ such that 
$\leb_0(V)\neq\leb_1(V)$. For each $n\in \omega$ define 
$\tau_n:(I\times \beta)\times
(\Borel_V\otimes \Power(\beta)) \to [0,1]$ as 
follows:
\begin{equation*}
  \tau_n((x,\eta),A) =
  \begin{cases}
    x\cdot \leb_0(A_0) &  \eta=0 \\ 
    \leb_0(A_{\alpha_n(\eta)}) &   \eta > 0, \, x\in (0,q_n) \\ 
    \leb_1(A_{\alpha_n(\eta)}) & \eta > 0, \, x\in [q_n,1)
  \end{cases}
\end{equation*}

Here $A_{\alpha_n(\eta)}$ is the $\alpha_n(\eta)$-section of $A$ 
defined before Lemma~\ref{lem:product_characterization}. As in the 
previous 
results, the definition of the kernels is motivated
by the process $\lmp{U}$. The two ``behaviors'' that we described,
in parallel and serial, are mimicked (by virtue of the definition of
$\alpha_n(\eta)$) at successor and limit ordinals $\eta$, respectively.

\begin{lemma}\label{kernel}
  For each $n \in \omega$, $\tau_n$ is a Markov kernel.
\end{lemma}
\begin{proof}
  It is clear that for  fixed $(x,\eta)$, the map
  $\tau_n((x,\eta),\cdot): \Borel_V \otimes \Power(\beta) \to [0,1]$
  is a subprobability measure. Let $A \in \Borel_V\otimes
  \Power(\beta)$; we want to show that  $\tau_n(\cdot,A)$ is
  measurable. For this, fix $q \in \mathbb{Q}\cap[0,1]$ and consider
  the set  
  $\{(x,\alpha)\in I\times \beta \mid  \tau_n((x,\alpha),A)<q\}$.
  By inspection of the definition of  $\tau_n$ in each fiber, we
  obtain that this set is the union of the following ones:
  \begin{gather*}
    \{(x,0) \mid x \cdot \leb_0(A_0)<q\},\quad  \bigcup_{\eta <
      \beta}\{(0,q_n)\times \{\eta+1\} \mid \leb_0(A_\eta)<q\},
    \\ \bigcup_{\eta < \beta}\{[q_n,1)\times \{\eta+1\} \mid
      \leb_1(A_\eta)<q\},\quad  \bigcup_{\lambda < 
      \beta}\{(0,q_n)\times
      \{\lambda\} \mid \leb_0(A_{\alpha_n(\lambda)})<q\},
      \\ \bigcup_{\lambda < \beta}\{[q_n,1)\times \{\lambda\} \mid
        \leb_1(A_{\alpha_n(\lambda)})<q\}.
  \end{gather*}
  Observe that each section of this union is either open or
  closed. Then it is measurable in the product space by
  Lemma~\ref{lem:product_characterization}.
\end{proof}
Hence we  have an LMP
\[
\lmp{S}(\beta) \defi 
\bigl( \Inter \times \beta,\;
\Borel_V \otimes \Power(\beta),\;
\{\tau_n\}_{n\in\omega} \bigr)
\]
for each $\beta\leq\omega_1$.
\begin{lemma}
  State bisimilarity $\sim_s$ on $\lmp{S}(\beta)$ is the identity.
\end{lemma}
\begin{proof}
  We will show by induction that for all  $1\leq\eta\leq\beta$, 
  $\sim_s\restriction (I\times \eta)$ is the identity. For the case
  $\eta=1$, we observe that if $(x,0)\neq(x',0)$, then for any $n \in
  \omega$
  \[
  \tau_n((x,0),I\times \beta)=x\neq  x'=\tau_n((x',0),I\times \beta)
  \]
  holds. Assume now that the result holds for $\eta$ and that
  $\eta+1\leq\beta$. By inductive hypothesis 
  $\sim_s\restriction  (I\times \eta)$ is the identity. 
  It is enough to consider states
  $(x,\alpha)\neq(x',\eta)$ for some $\alpha\leq\eta$. We analyze the
  case $\alpha< \eta$ first. The IH guarantees that every 
  $(r,\gamma)$ is only $\sim_s$-related to itself when $\gamma<\eta$; 
  since 
  $\alpha_n(\eta)<\eta$, it follows that the measurable set 
  $A(n)\defi(I\times \{\alpha_n(\eta)\})\cup (I\times\{\xi 
  \mid \eta\leq\xi<\beta\})$ 
  is an  element of $\Sigma(\sim_s)$ for every $n\in 
  \omega$. 
  If $\eta=1$ then $\alpha=0$, and in such case 
  $\tau_n((x,\eta),A(n))=\tau_n((x,\eta),I\times\beta)=1>x'
  =\tau_n((x',0),A(n))$
  holds for any  $n\in\omega$. For $\eta>1$, there exists $n\in
  \omega$ such that  $\alpha_n(\eta)\neq
  \alpha_n(\alpha)$ and we have
  \[
  \tau_n((x,\eta),A(n)) 
  =\leb_i(A(n)_{\alpha_n(\eta)})=\leb_i(I)=1\neq 0=
  \tau_n((x',\alpha),A(n)).
  \]
  Suppose now that $\eta=\alpha$ and $x\neq x'$; without loss of
  generality we can choose $n \in \omega$ such that $x<q_n<x'$. 
  Again, since
  $\alpha_n(\eta)=\alpha_n(\alpha)<\eta$, the inductive hypothesis 
  guarantees that the set 
  $A\defi  (V\times\{\alpha_n(\eta)\}) \cup (I \times \{\xi \mid
  \eta\leq\xi<\beta\})$ is in  $\Sigma(\sim_s)$. But 
  then
  \[
  \tau_n((x,\eta),A) =\leb_0(A_{\alpha_n(\eta)})=\leb_0(V)\neq
  \leb_1(V)=\leb_1(A_{\alpha_n(\alpha)})=\tau_n((x',\alpha),A).
  \] 
  This shows that the claim is also true for 
  $\eta+1$. Finally, assume  $\lambda$ is a limit ordinal. Since
  $\sim_s\restriction(I\times\lambda)=
  \bigcup_{\eta<\lambda}\sim_s\restriction(I\times\eta)$,
  the result follows easily from the IH.
\end{proof}

We now calculate a bound for the event bisimilarity,
${\sim_e}=\mathcal{R}(\sigma(\sem{\Logic}))$. We define 
\[
  \Lambda\defi 
  \{A\subseteq I\times\beta \mid A_0 
  \in\Borel(I) \wedge
  \forall \alpha>0 \, (A_\alpha \in \{\emptyset,I\})\}.
\]
\vspace*{-\baselineskip} %
\begin{lemma}\label{inclusion}
  $\sigma(\sem{\Logic}) \subseteq \Lambda$.
\end{lemma}
\begin{proof}
  It is clear that  $\Lambda$ is a  $\sigma$-algebra, now we verify 
  that it is stable. Let  $A \in \Lambda$
  and 
  $q\in \mathbb{Q}\cap \Inter $. By the same reasoning as in the proof
  of Lemma~\ref{kernel}, 
  $\{(x,\alpha)\in I\times \beta \mid \tau_n((x,\alpha),A)<q\} \in \Lambda$ 
  since it is Borel in the 0-fiber and, since
  $\leb_0(A_\eta)=\leb_1(A_\eta)$ for all  $\eta>0$, the remaining 
  sections
  are either $\emptyset$ or $I$. From this it follows that $\Lambda$
  is stable.
  Since $\sigma(\sem{\Logic})$ is the least
  stable $\sigma$-algebra by 
  Theorem~\ref{thm:sigma-logic-smallest-stable}, then
  $\sigma(\sem{\Logic})\subseteq\Lambda$.
\end{proof}
This bound is rather close to $\sigma(\sem{\Logic})$, as the 
following result
shows.
\begin{lemma}\label{lem:fibers-in-event-bisim}
  For all $\alpha<\beta$, $I\times \{\alpha\} \in 
  \sigma(\sem{\Logic})$.
\end{lemma}
\begin{proof}
  For the case $\alpha=0$ we observe that for any  $n\in \omega$
  $\{(x,\eta)\mid
  \tau_n((x,\eta),I\times\beta)<q\}=(0,q)\times \{0\}$, and this set
  is in  $\sigma(\sem{\Logic})$ because this $\sigma$-algebra is 
  stable. If we choose 
  $q=1$ we obtain
  the first case.
	
  Now assume that for a  given $\eta<\beta$, $I\times \{\alpha\} \in
  \sigma(\sem{\Logic})$ for all $\alpha<\eta$. For a fixed $n\in 
  \omega$, since
  $\alpha_n(\eta)<\eta$, the following set is in  
  $\sigma(\sem{\Logic})$:
  $\{(x,\xi)\mid \tau_n((x,\xi),I\times
  \{\alpha_n(\eta)\})<q\}=\bigcup_{\alpha < \beta}\{I\times
  \{\alpha\}\mid \alpha_n(\alpha)\neq\alpha_n(\eta)\}=I\times
  (\{\eta\}\cup
  \{\alpha<\beta\mid\alpha_n(\alpha)=\alpha_n(\eta)\})^c$. By taking 
  complements in the right hand side we obtain $A(n)\defi
  I\times(\{\eta\}\cup\{\alpha<\beta\mid\alpha_n(\alpha)
  =\alpha_n(\eta)\})
  \in \sigma(\sem{\Logic})$. But $I \times
  \{\eta\}=\bigcap_{n\in\omega}A(n) \in \sigma(\sem{\Logic})$ since 
  for  $\alpha
  \neq \eta$ there exists  $n\in\omega$ such that 
  $\alpha_n(\alpha)\neq\alpha_n(\eta)$. This completes the inductive 
  step.
\end{proof}

The next lemma gives some information about the  $\sigma$-algebra
$\G^\xi(\Lambda)$. Given $\eta<\beta$ and a    $\sigma$-algebra
$\mathcal{A}$, we will denote by $\mathcal{A}|_\eta$ the
restriction $\mathcal{A}\restriction I\times\{\eta\}$,
i.e.\ the $\sigma$-algebra of  $\eta$-sections of  elements of 
$\mathcal{A}$.

\begin{lemma}\label{lem:G-eta-restrict}
  If $\eta$ satisfies $\beta > \eta\geq \xi$, then
  $\G^{\xi}(\Lambda)|_\eta\subseteq\Borel(I)$. Also, 
  $\G^{\xi}(\Lambda)|_{\zeta+1}$ is trivial whenever 
  $\xi<\zeta+1<\beta$.
\end{lemma}
\begin{proof}
  By induction on $\xi$. If $\xi=0$,
  then $\G^0(\Lambda)=\Lambda$ and by its  definition,
  $\Lambda|_\eta\subseteq\Borel(I)$. Now suppose that the result 
  holds
  for  $\xi\geq 0$ and take $\eta\geq  \xi+1$.
  Let   $A\in
  \G^{\xi+1}(\Lambda)=\G(\G^\xi(\Lambda))=
  \Sigma(\mathcal{R}^{T}(\G^\xi(\Lambda)))$,
  i.e.\ $A \in \Borel_V\otimes \Power(\beta)$ is 
  $\mathcal{R}^{T}(\G^\xi(\Lambda))$-closed and
  therefore it is closed under this relation in each fiber. We
  aim to prove that  $A_\eta$ is Borel. We distinguish two cases: If
  $\eta$ is $\zeta+1$ for some $\zeta$, then $\zeta\geq\xi$ and by 
  the IH
  $\G^\xi(\Lambda)|_{\zeta}$ consists of Borel sets. Hence,
  for any set $D \in \G^\xi(\Lambda)$, any $n\in \omega$, and
  any $x\in I$, $\tau_n((x,\eta),D)=\leb(D_{\eta-1})$. In 
  consequence, if 
  $A_\eta \neq  \emptyset$ then it must be the case that
  $A_\eta=I$ because this set is 
  $\mathcal{R}^{T}(\G^\xi(\Lambda))$-closed. Moreover, the second
  claim in the statement of the Lemma follows.
  
  If  $\eta$ is a limit  ordinal we cannot argue as before because to 
  determine the  relation $\mathcal{R}^{T}(\G^\xi(\Lambda))$ we need
  to know the sections  $D_{\alpha_n(\eta)}$ of the elements in
  $\mathcal{G}^\xi(\Lambda)$, and it might be the case that
  $\alpha_n(\eta)<\xi$. Nevertheless,
  $\{n\in \omega \mid  \alpha_n(\eta)<\xi\}$ is finite and for the
  rest of  the naturals
  $m$, $\tau_m$ does not distinguish points in $A_\eta$ since by the
  IH $D_{\alpha_n(\eta)}$ is Borel. Now, if
  $\alpha_n(\eta)<\xi$, $\tau_n$ can only distinguish points between
  $[0,q_n)$ and $[q_n,1]$; in consequence, 
  $\G^{\xi+1}(\Lambda)|_\eta$ is the 
  $\sigma$-algebra generated by such  intervals, which is clearly   
  included in the Borel $\sigma$-algebra. This finishes the case  
  $\xi+1$.

  For case of $\xi$ limit, the IH ensures that for all
  $\gamma < \xi$ and $\eta\geq\gamma$,
  $\G^{\gamma}(\Lambda)|_\eta\subseteq\Borel(I)$. Hence, if
  $\eta\geq\xi$ then  $\eta>\gamma$ and this yields
  $\G^{\xi}(\Lambda)|_\eta=\sigma\bigl(\bigcup_{\gamma<\xi}
  \G^\gamma(\Lambda)\bigr)|_\eta=
  \sigma(\bigcup_{\gamma<\xi}\G^\gamma(\Lambda)|_\eta)
  \subseteq \Borel(I)$.
\end{proof}

\begin{corollary}\label{cor:zhou-geq-beta}
  If $\beta>\eta+1\geq\xi$, then $\Op^\xi(\sim_e)\restriction I\times 
  \{\eta+1\}$ is the total relation. In 
  consequence $\Zh(\lmp{S}(\beta))\geq \beta$ if $\beta$ is limit, 
  and $\Zh(\lmp{S}(\zeta+1))\geq\zeta$ for all $\zeta$.
\end{corollary}
\begin{proof}
  Combining the inclusion in Lemma~\ref{inclusion} 
  with the monotonicity
  of  $\G$, we see that
  $\Sigma_\xi=\G^\xi(\sigma(\sem{\Logic}))\subseteq \G^\xi(\Lambda)$ 
  and by 
  Lemma~\ref{lem:G-eta-restrict} 
  $\Sigma_\xi|_{\eta}$ also  consists of  Borel sets if
  $\eta\geq\xi$. As a consequence,
  $\mathcal{O}^\xi(\sim_e)=R_\xi\supseteq\mathcal{R}^T(\Sigma_\xi)$
  restricted to  $I\times\{\eta+1\}$ is the total relation because the
  measures $\leb_i$ cannot distinguish points if the allowed sets are 
  Borel. For the last assertion, take $\xi=\eta+1$ for any 
  $\xi<\beta$ in the limit case. For the second case, take 
  $\xi=\eta=\zeta-1$ if $\zeta$ is not limit. Otherwise, suppose by 
  way of a contradiction that $\Zh(\lmp{S}(\zeta+1))<\zeta$ and 
  choose any $\eta$ such that $\Zh(\lmp{S}(\zeta+1))<\eta<\zeta$, 
  then  
  $\sim_s\restriction 
  I\times\{\eta+1\}=\Op^{\Zh(\lmp{S}(\zeta+1))}(\sim_e)\restriction 
  I\times\{\eta+1\}$ is the total relation. This is a contradiction 
  because $\sim_s$ is the identity.
\end{proof}

The following lemma will provide a more detailed analysis
of the relations $R_\alpha$.

\begin{lemma}\label{lem:V_times_alpha_Sigma_R_alpha}
  For all $\alpha<\beta$,  
  $R_\alpha\restriction I\times
  (\alpha+1)$ is the identity and $V\times \{\alpha\}$ is in 
  $\Sigma(R_\alpha)$.
\end{lemma}
\begin{proof}
  For the case $\alpha=0$, we note that $I\times\beta \in 
  \sigma(\sem{\Logic})$, 
  hence, for any $q\in\Q$, $\{(x,\eta)\mid \tau_n((x,\eta),I\times 
  \beta)<q\} \in 
  \sigma(\sem{\Logic})$. Given that $\tau_n((s,0),I\times\beta)=s 
  \cdot \leb_0(I) 
  = s < 1$, if $s<t$, then for any $q\in \Q$ between $s$ and 
  $t$ we have $(s,0)\in\{(x,\eta)\mid 
  \tau_n((x,\eta),I\times\beta)<q\}\mathrel{\cancel{\ni}} (t,0)$. 
  Then, 
  $((s,0),(t,0))\notin{\sim_e}=\Rel(\sigma(\sem{\Logic}))$. This 
  shows that  
  $R_0={\sim_e}$ is the identity on $I \times \{0\}$. Moreover, if 
  $\eta>0$ and $x \in I$ then for any
  $n$, $\tau_n((x,\eta),I\times \beta)=1$ and hence $((s,0),(x,\eta))
  \notin {\sim_e}$. As a consequence, the set $V\times
  \{0\} \in \Borel_V \otimes \Power(\beta)$ is ${\sim_e}$-closed. Note
  that, since ${\sim_e}\supseteq R_\alpha$ for any $\alpha$, the 
  previous sentence
  shows that the $R_\alpha$-class of a point $(s,0)$ is the
  singleton $\{(s,0)\}$.
  
  Assume now that $\alpha+1<\beta$ and the result holds for  
  $\alpha$. Thanks to the inclusion
  $R_{\alpha+1}\subseteq R_\alpha$, the IH ensures that
  $R_{\alpha+1}\restriction I\times (\alpha+1)$ is the identity 
  relation
  and the set $V\times \{\alpha\}$ is in
  $\Sigma(R_\alpha)$. If $s < t$ we choose $n\in \omega$ such that 
  $s <  q_n < t$ and thus 
  $\tau_n((s,\alpha+1),V\times\{\alpha\})=\leb_0(V)\neq
  \leb_1(V)=\tau_n((t,\alpha+1),V\times\{\alpha\})$. 
  Moreover, the same set  $V\times \{\alpha\}$ serves as a test to
  distinguish points in the  $(\alpha+1)$-section and the previous 
  ones:
  $\tau_n((s,\alpha+1),V\times\{\alpha\})>0
  =\tau((t,\eta),V\times\{\alpha\})$
  for any $\eta<\alpha+1$. Therefore 
  $R_{\alpha+1}=\mathcal{R}^T\Sigma(R_\alpha)\restriction I\times
  (\alpha+2)$ is also the identity. To show that  $V \times
  \{\alpha+1\} \in \Sigma(R_{\alpha+1})$ it is enough to prove that 
  $((s,\alpha+1),(t,\eta)) \notin R_{\alpha+1}$ if $\alpha+1 <
  \eta$. For this we use the same  $V \times \{\alpha\}$
  provided by the IH. If  $\eta$ is a successor ordinal,
  for any $n$  $\tau_n((t,\eta),V \times \{\alpha\})=0$ holds and 
  if  $\eta$ is limit, we choose $n\in \omega$ such that
  $\alpha_n(\eta)\neq \alpha$ and again we obtain
  $\tau_n((t,\eta),V \times \{\alpha\})=0$.
  
  It remains to check the limit case. Assume that $\lambda<\beta$ is 
  a limit ordinal and the result holds for
  every $\alpha<\lambda$. Since
  $R_\lambda=\bigcap_{\alpha<\lambda}R_\alpha\subseteq R_\alpha$, 
  $R_\lambda\restriction I \times (\alpha+1)\subseteq
  R_\alpha\restriction I \times (\alpha+1)$ must be the identity by
  the IH. From this we conclude that  $R_\lambda\restriction I
  \times \lambda$ is the identity relation. If $s<t$, let $q_n \in \Q$
  such that  $s<q_n<t$. By IH,  $V\times
  \{\alpha_n(\lambda)\} \in \Sigma(R_{\alpha_n(\lambda)})$ and 
  therefore
  the inequality  $\tau_n((s,\lambda),V\times
  \{\alpha_n(\lambda)\})\neq\tau_n((t,\lambda),V\times
  \{\alpha_n(\lambda)\})$ allows us to conclude that
  $((s,\lambda),(t,\lambda))\notin R_{\alpha_n(\lambda)+1}\supseteq
  R_\lambda$. Now, if  $\eta\neq \lambda$, we choose $n$ such that
  $\alpha_n(\eta)\neq \alpha_n(\lambda)$ and we have
  $\tau_n((s,\lambda),V\times
  \{\alpha_n(\lambda)\})>0=\tau_n((t,\eta),V\times
  \{\alpha_n(\lambda)\})$. As before
  $((s,\lambda),(t,\eta))\notin R_{\alpha_n(\lambda)+1}\supseteq
  R_\lambda$. It follows that  $R_\lambda\restriction I\times 
  (\lambda+1)$
  is the identity and $V\times \{\lambda\}$ is in $\Sigma(R_\lambda)$.
\end{proof}
\begin{corollary}\label{cor:zhou-leq-beta}
  $\Zh(\lmp{S}(\zeta+1))\leq \zeta$. If $\beta$ is a limit ordinal, 
  then 
  $\Zh(\lmp{S}(\beta))\leq \beta$.
\end{corollary}
\begin{proof}
  The first statement follows directly from 
  Lemma~\ref{lem:V_times_alpha_Sigma_R_alpha} by taking 
  $\alpha=\zeta$. For the second one, if $\beta$ is a limit ordinal 
  $\forall \alpha<\beta$ 
  $R_\alpha\supseteq R_\beta$, then if $R_\alpha\restriction 
  I\times(\alpha+1)$ equals the identity, $R_\beta\restriction 
  I\times(\alpha+1)$ is also the identity. Therefore 
  $R_\beta\restriction I\times\beta$ is the identity and 
  $\Zh(S(\beta))\leq \beta$.
\end{proof}
\begin{theorem}\label{distancia-contable}
  For a limit $\beta\leq\omega_1$,  $\Zh(\lmp{S}(\beta))=\beta$; and
  if  $\beta<\omega_1$, $\Zh(\lmp{S}(\beta+1))=\beta$.\qed
\end{theorem}
From this we have another proof of Corollary~\ref{cor:Zh-seplmp},
since it is elementary to check that for countable $\beta$,
$\lmp{S}(\beta)$ is separable.

We close this section by studying the possible dependency of the value
of $\Zh(\seplmp)$ under different set-theoretical hypotheses
(consistent relative to Zermelo-Fraenkel set theory).

Assuming the Continuum Hypothesis ($\CH$) we have
$2^{\aleph_0}=\aleph_1<2^{\aleph_1}$ and in consequence the space
$(I\times\omega_1,\Borel_V\otimes\Power(\omega_1))$ is not 
separable. 
Indeed, there are $2^{\aleph_1}=\left\vert{\Power(\omega_1)}\right\vert$
subsets of  $\omega_1$ and each of them defines a different measurable
set in the product $\sigma$-algebra, but the  cardinal of a countably generated
$\sigma$-algebra is at most $2^{\aleph_0}$.

On the other hand, if we assume \emph{Martin's Axiom} ($\mathit{MA}$)
and the negation of $\mathit{CH}$, any uncountable subset $X\subseteq\mathbb{R}$
with cardinality less than $2^{\aleph_0}$ is a \emph{$Q$-set} (see
e.g.\ Miller \cite{miller}),
viz.\ one
such that all of its subsets are relative $G_\delta$. If we choose $X$
such that $\card{X}=\aleph_1$, the relative topology has a countable
base (the relativization of such a base for $\mathbb{R}$) and it
generates $\Power(X)$ as a $\sigma$-algebra since $X$ is a
$Q$-set. Then, $\Power(X)\cong\Power(\omega_1)$ is separable and in
consequence $(I\times\omega_1,\Borel_V\otimes\Power(\omega_1))$, the
state space of $\lmp{S}(\omega_1)$, also is. In
this context the bound in Corollary~\ref{cor:Zh-seplmp} is attained.

%% file: conclusion.tex
\section{Conclusion}\label{sec:conclusion}

The Zhou ordinal provides a measure of the failure of the
Hennessy-Milner property on labelled Markov processes over general
measurable spaces. The general study of this ordinal on processes over
separable metrizable spaces has opened several questions.

The proof of Theorem~\ref{th:zhou-is-limit} shows in particular that
given a process $\lmp{S}$ with $\Zh(\lmp{S})\geq\alpha$, we can
construct a second one $\lmp{S'}$ with 
$\Zh(\lmp{S'})=\alpha+1$,
whenever $\alpha$ is a \emph{successor} ordinal. But we actually do
not know how to obtain this result for general $\alpha$. This is even
clearer for $\alpha=0$: The construction of the initial counterexample
from \cite{Pedro20111048}
does not follow the pattern of what we do in the passage from
$\alpha+1$ to $\alpha +2$.
The same happens with the proof of the
Theorem~\ref{th:cf-Zhou-gt-omega}. Given
$\{\lmp{S}_\alpha\}_{\alpha<\beta}$ with $\beta$ limit, the second
general question is to construct in a natural way some
$\lmp{S}_\beta$ such that  
$\Zh(\lmp{S}_\beta)=\sup_{\alpha<\beta} \Zh(\lmp{S}_\alpha)$.
For the case of countable $\beta$, we obtain a process $\lmp{T}$ with 
$\Zh(\lmp{T})\geq \sup_{\alpha<\beta} \Zh(\lmp{S}_\alpha)$ (actually,
strictly greater).

It is to be noted that the last inequality can be
upgraded to an equality by passing to an appropriate quotient. We know
how to perform this construction to get a process with Zhou ordinal
a limit  $\beta$ from one with a larger ordinal, but this needs
further study in general. Another avenue to pursue is the 
characterization of event bisimilarity on the processes 
$\lmp{S}(\alpha)$. It can be proved that our proposed bound $\Lambda$ 
is indeed equal to $\sigma(\sem{\Logic})$ for countable $\alpha$; 
also $\sigma(\sem{\Logic})$ is always countably generated. But 
consistently, $\Lambda$ on $\lmp{S}(\omega_1)$ is not.

As the main open question, we did not settle if 
$\Zh(\seplmp)$ is
actually a (regular) cardinal. An early conjecture was that
$\Zh(\seplmp) =\omega_1$ (unconditionally), but this is now counterintuitive
in view of the existence of a separable LMP with such ordinal under $\MA+\neg\CH$.
It is also to be noted that if we were able  to pass from any 
process with ordinal $\alpha$ to one with ordinal $\alpha+1$, the
same set-theoretical assumptions would let us conclude 
$\Zh(\seplmp) \geq \omega_1\cdot 2$ by 
Theorem~\ref{th:cf-Zhou-gt-omega}. To put it in focus, (consistently)
finding a 
separable $\lmp{S}$ with $\Zh(\lmp{S})\geq \omega_1+1$ is the next
question to address.

%% file: zhou_ordinal.bbl
\providecommand{\noopsort}[1]{}